\documentclass[10pt, conference, letterpaper]{IEEEtran}
\usepackage{cite}
\ifCLASSINFOpdf
\usepackage[pdftex]{graphicx}
\else
\usepackage[dvips]{graphicx}
\fi
\usepackage[cmex10]{amsmath}
\usepackage{algorithm}
\usepackage{algorithmic}
\usepackage{amsmath}
\usepackage{array}
\usepackage{framed}
\usepackage{bm}
\usepackage[tight,footnotesize]{subfigure}
\usepackage{multirow}
\usepackage{tikz}
\usetikzlibrary{calc}
\usepackage{amsmath,amsthm,amssymb,amsfonts}
\usepackage[paperwidth=8.5in, paperheight=11in, margin=0.8in]{geometry}
\usepackage{color}

\newtheorem{theorem}{Theorem}
\newtheorem{lemma}{Lemma}
\newtheorem{definition}{Definition}
\newtheorem{observation}{Observation}

\newtheorem{assumption}{Assumption}

\IEEEoverridecommandlockouts
\begin{document}

\title{Optimal Network Control in Partially-Controllable Networks}
\author{
\IEEEauthorblockN{Qingkai Liang and Eytan Modiano}
\IEEEauthorblockA{Laboratory for Information and Decision Systems\\Massachusetts Institute of Technology, Cambridge, MA}
\thanks{This work was supported by NSF Grant CNS-1524317 and by DARPA I2O and Raytheon BBN Technologies under Contract No. HROO l l-l 5-C-0097.}
}

\maketitle

\begin{abstract}
The effectiveness of many optimal network control algorithms (e.g., BackPressure) relies on the premise that all of the nodes are fully controllable. However, these algorithms may yield poor performance in a partially-controllable network where a subset of nodes are uncontrollable and use some unknown policy. Such a partially-controllable model is of increasing importance in real-world networked systems such as overlay-underlay networks. In this paper, we design optimal network control algorithms that can stabilize a partially-controllable network. We first study the scenario where uncontrollable nodes use a  queue-agnostic policy, and propose a low-complexity throughput-optimal algorithm, called Tracking-MaxWeight (TMW), which enhances the original MaxWeight algorithm with an explicit learning of the policy used by uncontrollable nodes. Next, we investigate the scenario where uncontrollable nodes use a queue-dependent policy and the problem is formulated as an MDP with unknown queueing dynamics. We propose a new reinforcement learning algorithm, called Truncated Upper Confidence Reinforcement Learning (TUCRL), and prove that TUCRL achieves tunable three-way tradeoffs between throughput, delay and convergence rate.
\end{abstract}

\begin{tikzpicture}[remember picture, overlay]
\node at ($(current page.north) + (-3in,-0.5in)$) {Technical Report};
\end{tikzpicture}

\section{Introduction}
Optimal network control  has been an active area of research  for more than thirty years, and many efficient routing algorithms have been developed over the past few decades, such as the well-known throughput-optimal BackPressure routing algorithm \cite{tassiulas}. The effectiveness of these algorithms usually relies on the premise that all of the nodes in a network are fully controllable. Unfortunately, an increasing number of real-world networked systems are only \emph{partially controllable}, where a subset of nodes are not managed by the network operator and use some unknown network control policy, such as overlay-underlay networks.

An overlay-underlay network consists of overlay nodes and underlay nodes \cite{jones-overlay, paschos-overlay, rai-overlay}. The overlay nodes can implement state-of-the-art algorithms while the underlay nodes are uncontrollable and use some unknown protocols (e.g., legacy protocols). 
Figure \ref{fig:underlay_example} shows an overlay-underlay network where the communications among overlay nodes rely on the uncontrollable underlay nodes. 
Overlay networks have been used to improve the  capabilities of computer networks for a long time (e.g., content delivery \cite{overlay-4}).
\begin{figure}[ht!]
\begin{center}
\includegraphics[width=2.3in]{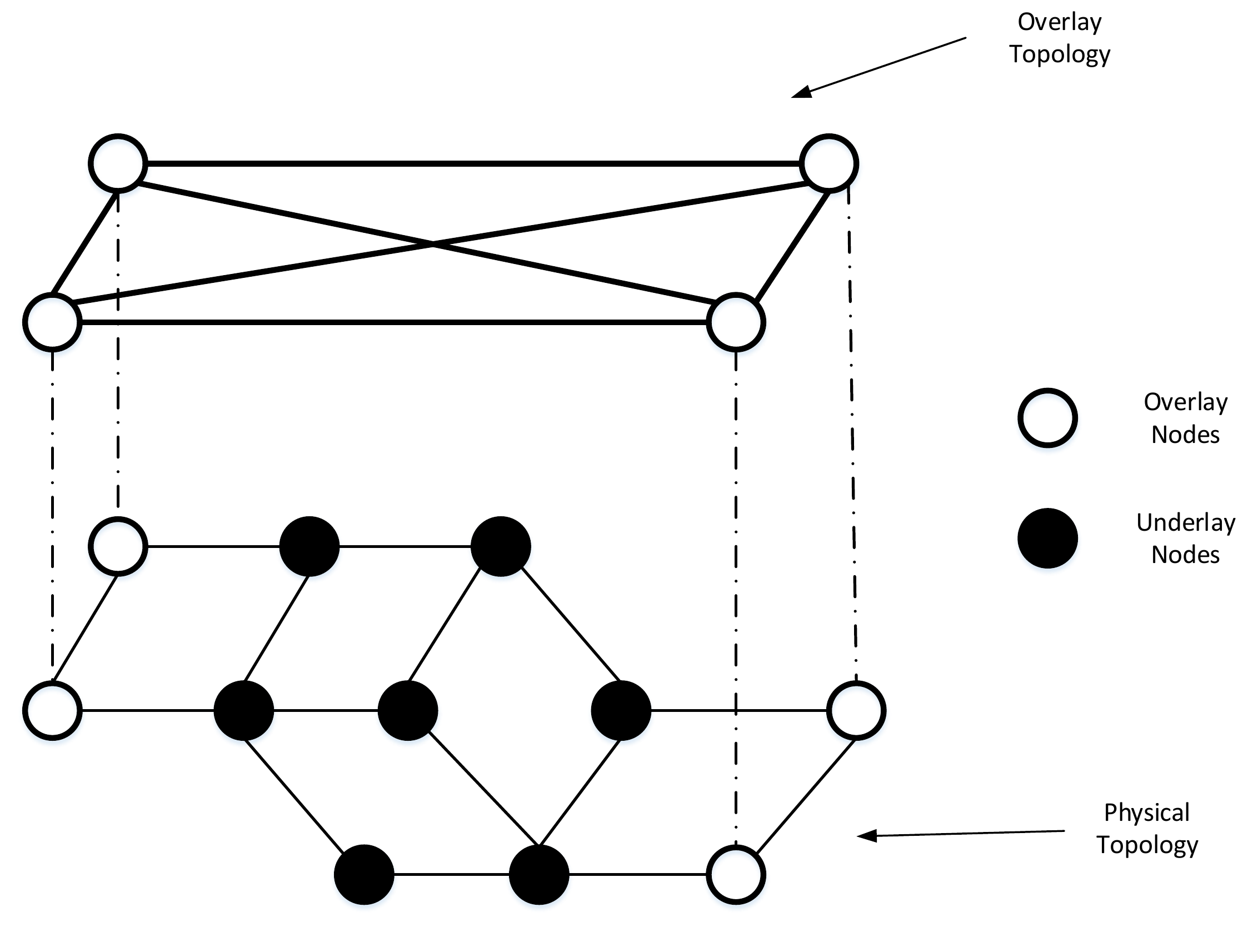}
\caption{An example of overlay-underlay networks.}
\label{fig:underlay_example}
\end{center}
\end{figure}


Due to the unknown behavior of uncontrollable nodes, the existing routing algorithms may yield poor performance in a partially-controllable network. For example, Figure \ref{fig:example} shows an example where the well-known Backpressure routing algorithm \cite{tassiulas} fails to deliver the maximum throughput when some nodes are uncontrollable. In particular, uncontrollable node 3 adopts a policy that does not preserve the flow conservation law such that its backlog builds up, but uncontrollable node 2 hides this backlog information from node 1. As a result, if node 1 uses Backpressure routing, it always transmits packets to node 2, although these packets will never be delivered. A smarter algorithm should be able to learn the behavior of the uncontrollable nodes such that node 1 only sends packets along route $1\rightarrow 5\rightarrow 4$.

As a result, it is important to develop new network control algorithms that achieve consistently good performance in a partially-controllable environment. In this paper, we study efficient network control algorithms that can stabilize a partially-controllable network whenever possible. In particular, we consider two scenarios.

First, we investigate the scenario where uncontrollable nodes use a \emph{queue-agnostic} policy, which captures a wide range of practical network protocols, such as shortest path routing  (e.g., OSPF, RIP), multi-path routing  (e.g., ECMP) and randomized routing algorithms. In this scenario, we propose a low-complexity throughput-optimal algorithm, called Tracking-MaxWeight (TMW), which enhances the original MaxWeight algorithm \cite{tassiulas} with an explicit learning of the policy used by uncontrollable nodes.

Second, we study the scenario where uncontrollable nodes use a \emph{queue-dependent} policy, i.e., the action taken by uncontrollable nodes relies on the observed queue length vector (e.g., Backpressure routing). In this scenario, we show that the queueing dynamics become unknown and no longer follow the classic Lindley recursion \cite{lindley}, which makes the problem fundamentally different from the traditional network optimization framework: we not only need to know how to perform optimal network control but also need to learn the queueing dynamics in an efficient way. We formulate the problem as a Markov Decision Process (MDP) with unknown dynamics, and propose a new reinforcement learning algorithm, called Truncated Upper Confidence Reinforcement Learning (TUCRL), that is shown to achieve network stability under mild conditions.

\begin{figure}[ht!]
\begin{center}
\includegraphics[width=2.2in]{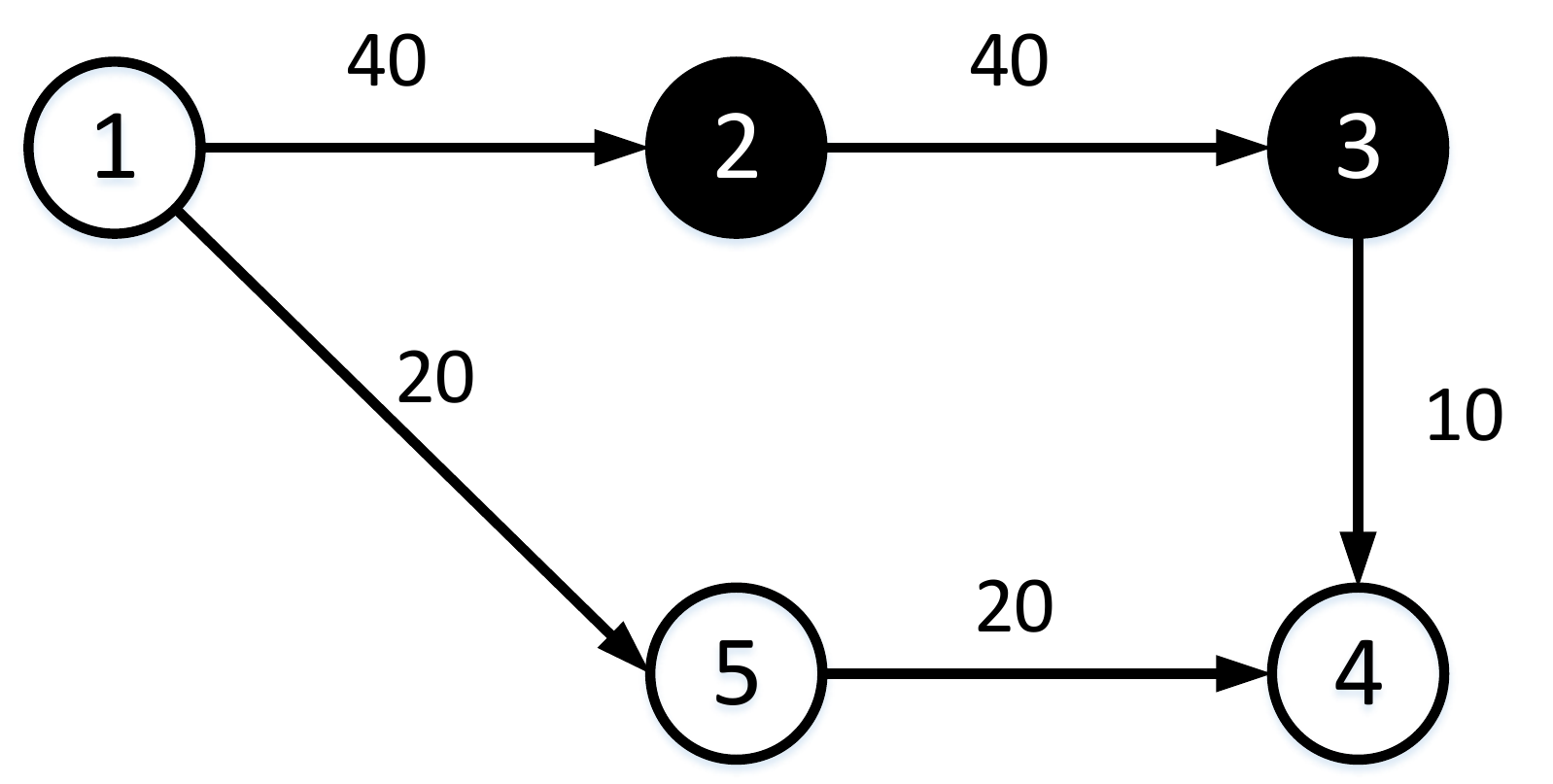}
\caption{Counterexample where the well-known Backpressure routing algorithm fails to deliver the maximum throughput in a partially-controllable network. The number next to each link is its capacity. Each node can transmit only to one of its neighbors in each time slot. There is only one flow: $1\rightarrow 4$ (at rate 20). Black nodes are uncontrollable nodes. Node $2$ transmits any packet it received to node $3$ at full rate, so that its queue length is always zero; node $3$ adopts a non-work-conserving policy that holds any packet it received. When node 1 uses Backpressure routing, it always transmits packets to node 2 since its queue length is always zero, which hides the fact that backlog builds up at node 3.}
\label{fig:example}
\end{center}
\end{figure}

\subsection{Related Work}
Most of the existing works on network optimization in a partially-controllable environment are in the context of overlay-underlay networks.  An important feature of overlay-underlay networks is that underlay nodes are not controllable and may adopt arbitrary (unknown) policies. The objective is to find efficient control policies for the controllable overlay nodes in order to optimize certain performance metrics (e.g., throughput). In \cite{jones-overlay}, the authors showed that the well-known BackPressure algorithm \cite{tassiulas}, which was shown to be throughput-optimal in a wide range of scenarios, may lead to a loss in throughput when used in an overlay-underlay setting, and proposed a heuristic routing algorithm for overlay nodes called Overlay Backpressure Policy (OBP).  An optimal backpressure-type routing algorithm for a special case, where the underlay paths do not overlap with each other, was given in \cite{paschos-overlay}. Recently, \cite{rai-overlay} showed that the overlay routing algorithms proposed in \cite{jones-overlay}\cite{paschos-overlay} are not throughput-optimal in general, and developed the Optimal Overlay Routing Policy (OORP) for overlay nodes. However, all of the existing overlay routing algorithms \cite{jones-overlay, paschos-overlay, rai-overlay} impose very stringent assumptions about the behavior of underlay nodes. In particular, the underlay nodes are required to use fixed-path routing (e.g., shortest-path routing) and maintain stability  whenever possible, which fails to account for many important underlay policies (e.g., underlay nodes may use multi-path routing). 



In terms of technical tools, our work leverages techniques from reinforcement learning, since a partially-controllable network with queue-dependent uncontrollable policy can be formulated as an MDP with unknown dynamics. Over the past few decades, many reinforcement learning algorithms have been developed, such as Q-learning \cite{rl-book}, actor critic \cite{ac} and policy gradient \cite{reinforce}. Recently, the successful applications of deep neural networks in reinforcement learning algorithms have produced many deep reinforcement learning algorithms such as DQN \cite{dqn}, DDPG \cite{ddpg} and TRPO \cite{trpo}. However, most of these methods are heuristic-based and do not have any performance guarantees. Among the existing reinforcement learning algorithms there are a few that do provide good performance bounds, such as model-based reinforcement learning algorithms UCRL \cite{UCRL1,UCRL2} and PSRL \cite{PSRL1,PSRL2}. Unfortunately, these algorithms require that the size of the state space be relatively small, which cannot be directly applied in our context since the state space (i.e., queue length space) contains countably-infinite states. 
In this work, we combine the UCRL algorithm with a queue truncation technique and propose the Truncated Upper Confidence Reinforcement Learning (TUCRL) algorithm that has good performance guarantees even with countably-infinite queue length space. 

\subsection{Our Contributions}
In this paper, we investigate optimal network control  for a partially-controllable network. Whereas existing works (e.g., \cite{jones-overlay, paschos-overlay, rai-overlay})  imposed very stringent assumptions about the behavior of uncontrollable nodes  for analytical tractability, this is the first work that establishes stability results under the generalized partially-controllable network model. In particular, we develop two  network control algorithms.

First, we develop a low-complexity Tracking-MaxWeight (TMW) algorithm that is guaranteed to achieve network stability if uncontrollable nodes adopt queue-agnostic policies. The Tracking-MaxWeight algorithm enhances the original MaxWeight algorithm \cite{tassiulas} with an explict learning of the policy used by uncontrollable nodes. 

Next, we propose a new reinforcement learning algorithm (i.e., the TUCRL algorithm) in the more challenging scenario where uncontrollable nodes may use queue-dependent policies. It combines the state-of-art model-based UCRL algorithm  \cite{UCRL1,UCRL2} with a queue truncation technique to overcome the problem with countably-infinite queue length space. We prove that  TUCRL  achieves network stability by dropping a negligible fraction of packets. We also show that the TUCRL algorithm maintains a three-way tradeoff between delay, throughput and convergence rate.

\section{System Model}\label{sec:model}
Consider a networked system with $N$ nodes (the set of all nodes is denoted by $\mathcal{N}$). There are $K$ flows in the network and each node $i$ maintains a queue for buffering undelivered packets for each flow $k$. As a result, there are $NK$ queues in the network, and we denote by $\mathbf{Q}(t)$ the queue length vector at the beginning of time slot $t$, where its element $Q_{ik}(t)$ represents the queue length for flow $k$ at node $i$.

Let  $\omega_t$ be the \emph{network event} that occurs in slot $t$, which includes information about the current network parameters, such as a vector of channel conditions for each link and a vector of exogenous arrivals to each queue.
We assume that the sequence of network events $\{\omega_t\}_{t\ge 0}$ follow a stationary stochastic process.  In particular, the vector of exogenous packet arrivals is denoted by $\mathbf{a}(\omega_t)=\{a_{ik}(t)\}_{i,k}$, where $a_{ik}(t)$ is the number of exogenous arrivals to queue $(i,k)$ in slot $t$. Denote by $\lambda_{ik} = \mathbb{E}[a_{ik}(t)]$ the expected exogenous packet arrival rate to queue $(i,k)$ in steady state.

At the beginning of each time slot $t$, after observing the current network event $\omega_t$ and the current queue length vector $\mathbf{Q}(t)$, each node $i$ needs to make a routing decision $f_{ijk}(t)$ indicating the offered transmission rate for flow $k$ over link $i\rightarrow j$. The corresponding network routing vector is denoted by $\mathbf{f}(t)=\{f_{ijk}(t)\}_{i,j,k}$. 

There are two types of nodes in the network: controllable nodes (the set of controllable nodes is denoted by $\mathcal{C}$) and uncontrollable nodes (the set of uncontrollable nodes is denoted by $\mathcal{U}$). The network operator can only control the routing behavior for controllable nodes while the routing actions taken by uncontrollable nodes cannot be regulated and are only observable at the end of each time slot. In this case, the network routing vector $\mathbf{f}(t)$ can be decomposed into two parts: $\mathbf{f}(t) = (\mathbf{f}^{c}(t), \mathbf{f}^{u}(t))$. Here, $\mathbf{f}^{c}(t)=\{f_{ijk}(t)\}_{i\in\mathcal{C}}$ represents the routing decisions made by controllable nodes (referred to as the \textbf{controllable action}) and $\mathbf{f}^{u}(t)=\{f_{ijk}(t)\}_{i\in\mathcal{U}}$ corresponds to the routing decisions made by uncontrollable nodes (referred to as the \textbf{uncontrollable action}). The routing vectors $\mathbf{f}^{c}(t)$ and $\mathbf{f}^{u}(t)$ are constrained within some action spaces $\mathcal{F}^c_{\omega_t}$ and $\mathcal{F}^u_{\omega_t}$, respectively, that may depend on the current network event $\omega_t$, respectively. The action space for all nodes is denoted by $\mathcal{F}_{\omega_t}=\mathcal{F}_{\omega_t}^c \cup \mathcal{F}_{\omega_t}^u$. The action space can be used to specify routing constraints (e.g., the total transmission rate over each link should not exceed its capacity)  or describe  scheduling constraints (e.g.,  each node can only  transmit to one of its neighbors in each time slot).

Note that when there is not enough backlog to transmit, the actual number of transmitted packets may be less than the offered transmission rate. In particular, we denote by $\widetilde{f}_{ijk}(Q_{ik}(t))$ (or simply $\widetilde{f}_{ijk}(t)$ if the context is clear) the actual number of transmitted packets in flow $k$ over link $i\rightarrow j$ in slot $t$ under the current queue length $Q_{ik}(t)$. Clearly, we have $\widetilde{f}_{ijk}(Q_{ik}(t)) \le \min\{f_{ijk}(t), Q_{ik}(t)\}$.
We further assume that the routing decision can always be chosen to respect the backlog constraints (but the actual actions may not necessarily be queue-respecting). This can be done simply by never attempting to transmit more data than we have. Under such notations, the queuing dynamics are given by
\[
\begin{split}
\small
&Q_{ik}(t+1)\\
=&Q_{ik}(t) + a_{ik}(t) + \sum_{j\in\mathcal{N}}\widetilde{f}_{jik}(t) - \sum_{j\in\mathcal{N}}\widetilde{f}_{ijk}(t)\\
\le &\Big[Q_{ik}(t) + a_{ik}(t) + \sum_{j\in\mathcal{N}}f_{jik}(t) - \sum_{j\in\mathcal{N}}f_{ijk}(t)\Big]^+,
\end{split}
\]
where $[z]^+=\max\{z,0\}$. 
We also make the following boundedness assumption: the amount of exogenous arrivals and the offered transmission rate in each time slot are bounded by some constant $D$, i.e., 
\[
0\le a_{ik}(t)\le D,~~0\le f_{ijk}(t)\le D,~\forall i,j,k.
\]

A network control policy $\pi$ is a mapping from the observed network event $\omega$ and queue length vector $\mathbf{Q}$ to a feasible routing action. In particular, denote by $\pi_c:(\omega,\mathbf{Q})\mapsto \mathbf{f}^c$ a \textbf{controllable policy} and $\pi_u:(\omega,\mathbf{Q})\mapsto \mathbf{f}^u$ an \textbf{uncontrollable policy}. In this paper, we assume that the uncontrollable policy $\pi_u$ remains fixed over time  but is unknown to the network operator. Our objective is to find a controllable policy $\pi_c$ such that network stability can be achieved, as is defined as follows.
\begin{definition}\label{def:stability}
A network is rate stable if 
\[
\lim_{t\rightarrow\infty}\frac{\mathbb{E}[Q_{ik}(t)]}{t}=0,~\forall i,k.
\]
Rate stability means that the average arrival rate to each queue equals the average departure rate from that queue.
\end{definition}

\section{Queue-Agnostic Uncontrollable Policy}\label{sec:agnostic}
In this section, we consider the scenario where the uncontrollable policy is queue-agnostic, which simply observes the current network event $\omega_t$ and makes a routing decision $\mathbf{f}^u(t)\in\mathcal{F}_{\omega_t}^u$ as a stationary function only of $\omega_t$, i.e., $\pi_u:\omega_t\mapsto \mathbf{f}^u(t)$. In the stochastic network optimization literature, such a policy is also referred to as an \textbf{$\omega$-only policy} \cite{neely-sno}. Despite their simple form, $\omega$-only policies can capture a wide range of network control protocols in practice, such as  shortest-path routing protocols (e.g., OSPF, RIP) , multi-path routing protocols (e.g., ECMP) and randomized routing protocols.

Unfortunately, even under simple $\omega$-only uncontrollable policies, existing routing algorithms  may fail to stabilize the network. For example, as is illustrated in Figure \ref{fig:example}, the well-known Backpressure routing algorithm achieves low throughput when uncontrollable node uses queue-agnostic policies. 
In this example, the failure is due to the fact that some uncontrollable node uses a non-stabilizing policy that does not preserve flow conservation  but the Backpressure algorithm is not aware of this non-stabilizing behavior.


In this section, we propose a low-complexity  algorithm that learns the behavior of uncontrollable nodes and achieves network stability under any $\omega$-only uncontrollable policy. 

\subsection{Tracking-MaxWeight Algorithm}
Now we introduce an algorithm that achieves network stability  whenever uncontrollable nodes use an $\omega$-only policy. The algorithm is called Tracking-MaxWeight (TMW), which enhances the original MaxWeight algorithm \cite{tassiulas} with an explicit learning of the policy used by uncontrollable nodes.  Throughout this section, we let $\{\mathbf{f}^u(t)\}_{t\ge 0}$ be the sequence of routing actions that are actually executed by uncontrollable nodes.

The details of the TMW algorithm are presented in Algorithm \ref{alg:TMW}. In each slot $t$, the TMW algorithm  generates the routing actions $\mathbf{g}^c(t)=\{g_{ijk}(t)\}_{i\in\mathcal{C}}$ for controllable nodes and also produces an ``imagined" routing action $\mathbf{g}^u(t)=\{g_{ijk}(t)\}_{i\in\mathcal{U}}$ for uncontrollable nodes, by solving the optimization problem \eqref{eq:TMW}. With these calculated actions, the TMW algorithm then updates two virtual queues. The first virtual queue $\mathbf{X}(t)$ tries to emulate the physical queue $\mathbf{Q}(t)$ but assumes that the imagined uncontrollable action $\mathbf{g}^u(t)$ is applied (while the physical queue is updated using the true uncontrollable action $\mathbf{f}^u(t)$). The second virtual queue $\mathbf{Y}(t)$ tracks the cumulative difference between the imagined uncontrollable actions $\{\mathbf{g}^u(t)\}_{t\ge 0}$ and the actual uncontrollable actions $\{\mathbf{f}^u(t)\}_{t\ge 0}$. In particular, we use $\Delta_{ijk}(t)$ to measure the difference between the imagined routing action $g_{ijk}(t)$ and the true routing action $f_{ijk}(t)$ taken by uncontrollable node $i\in\mathcal{U}$, which is given by
\begin{equation}\label{eq:diff}
\Delta_{ijk}(t) = g_{ijk}(t) - \widetilde{f}_{ijk}(t),~\forall i\in\mathcal{U},
\end{equation}
where $\widetilde{f}_{ijk}(t)$ is the actual number of transmitted packets under the true routing action $f_{ijk}(t)$ given the current queue backlog $\mathbf{Q}(t)$.
Note that for each controllable node $i\in\mathcal{C}$, we simply set $\Delta_{ijk}(t)=0$.

The optimization problem \eqref{eq:TMW}  aims at maximizing a weighted sum of flow variables, which is similar to the optimization problem solved in the original MaxWeight algorithm \cite{tassiulas} except for the setting of weights. In the original MaxWeight algorithm, the weight is $W_{ijk}(t)=Q_{ik}(t)-Q_{jk}(t)$ corresponding to the physical queue backlog differential, while in the Tracking-MaxWeight algorithm the weight $W_{ijk}(t)=X_{ik}(t)-X_{jk}(t)-Y_{ijk}(t)$ accounts for both the backlog differential for virtual queue $\mathbf{X}(t)$ and the backlog of virtual queue $\mathbf{Y}(t)$. The derivation of  \eqref{eq:TMW} is based on the minimization of quadratic Lyapunov drift terms for the two virtual queues:
\[
\small
\begin{split}
\min_{\mathbf{g}(t)\in\mathcal{F}_{\omega_t}} ~~~ &\sum_{i,k} X_{ik}(t)\Big[a_{ik}(t)+\sum_{j}g_{jik}(t)-\sum_j g_{ijk}(t)\Big] \\
&+\sum_{i,j,k} Y_{ijk}(t)\Big(g_{ijk}(t) - \widetilde{f}_{ijk}(t)\Big),
\end{split}
\]
where the first term corresponds to the Lyapunov drift of virtual queue $\mathbf{X}(t)$ and the second term is the Lyapunov drift of virtual queue $\mathbf{Y}(t)$. Note that the minimization is done over controllable actions $\mathbf{g}^c(t)$ and  ``imagined" uncontrollable actions $\mathbf{g}^u(t)$.  Cleaning up irrelevant constants, i.e., $a_{ik}(t)$ and $ \widetilde{f}_{ijk}(t)$, and rearranging terms yield the optimization problem \eqref{eq:TMW}.

\begin{algorithm}[ht!]
\caption{Tracking-MaxWeight (TMW)}\label{alg:TMW}
\begin{algorithmic}[1]
\STATE In each slot $t$, observe the current network event $\omega_t$ and solve the following optimization problem to obtain the controllable action $\mathbf{g}^c(t)$ and the \emph{imagined} uncontrollable action $\mathbf{g}^u(t)$:
\begin{equation}\label{eq:TMW}
\small
\max_{\mathbf{g}(t)\in\mathcal{F}_{\omega_t}}~~~ \sum_{(i,j)}\sum_k g_{ijk}(t) W_{ijk}(t),
\end{equation}
where
\[
\small
W_{ijk}(t)=X_{ik}(t)-X_{jk}(t)-Y_{ijk}(t).
\]

\STATE  Controllable nodes execute the routing decision $\mathbf{g}^c(t)$.

\STATE Observe the true routing action $\mathbf{f}^u(t)$ taken by uncontrollable nodes and update virtual queues:
\[
\small
\begin{split}
&X_{ik}(t+1) = \Big[X_{ik}(t) + a_{ik}(t) + \sum_{j\in\mathcal{N}}g_{jik}(t) - \sum_{j\in\mathcal{N}}g_{ijk}(t)\Big]^+\\
&Y_{ijk}(t+1) = Y_{ijk}(t) + \Delta_{ijk}(t)
\end{split}
\]
where $\Delta_{ijk}(t)$ is defined in \eqref{eq:diff}.
\end{algorithmic}
\end{algorithm}

Next we show that Tracking-MaxWeight achieves  stability whenever uncontrollable nodes use an $\omega$-only policy and the network is \emph{within the stability region}, i.e., there exists a sequence of feasible routing vectors $\{\mathbf{f}^{c}(t)\}_{t\ge 0}$ for controllable nodes such that 
\begin{equation}\label{eq:loaded}
\lambda_{ik} +  \sum_{j\in\mathcal{N}} \widetilde{f}_{jik}-\sum_{j\in\mathcal{N}}\widetilde{f}_{ijk} \le 0,~\forall i,k,
\end{equation}
where
\[
\widetilde{f}_{ijk} = \lim_{T\rightarrow\infty}\frac{1}{T}\sum_{t=0}^{T-1} \mathbb{E}[\widetilde{f}_{ijk}(Q^*_{ik}(t))]
\]
is the long-term average actual flow transmission rate under $\{\mathbf{f}^c(t),\mathbf{f}^u(t)\}_{t\ge 0}$ and $\{\mathbf{Q}^*(t)\}_{t\ge 0}$ is the corresponding optimal queue length trajectory. In other words, \eqref{eq:loaded} requires that  flow conservation  should be preserved for every queue under the optimal controllable policy, otherwise no algorithm can stabilize the network. It is important to note that in \eqref{eq:loaded} the flow conservation law is with respect to the \emph{actual} transmissions  since an uncontrollable node may not preserve  flow conservation in terms of its offered transmissions (e.g., in Figure \ref{fig:example},  the offered incoming rate to node 3 is 40 while the offered outgoing rate from node 3 is 0). The only way to stabilize these nodes is by limiting the amount of backlog such that the actual endogenous arrivals to these nodes are smaller. The performance of Tracking-MaxWeight is given in the following theorem.

\begin{theorem}\label{thm:TMW}
When uncontrollable nodes use an $\omega$-only policy and the network is within the stability region, Tracking-MaxWeight achieves rate stability.
\end{theorem}
\begin{proof}
The proof first shows that the two virtual queues $\mathbf{X}(t)$ and $\mathbf{Y}(t)$ can be stabilized by the TMW algorithm by using the Lyapunov drift analysis. Then we prove that whenever the two virtual queues are stable, the physical queue $\mathbf{Q}(t)$ is also stable. See Appendix \ref{ap:TMW} for details.
\end{proof}

\section{Queue-Dependent Uncontrollable Policy}\label{sec:dependent}
The previous section investigated the scenario where uncontrollable nodes use a queue-agnostic policy (i.e., $\omega$-only policy). In this section, we study a more general case where the uncontrollable policy may be \emph{queue-dependent}, which can be used to describe many state-of-the-art optimal network control protocols. For example, the well-known Backpressure algorithm makes routing decisions based on the currently observed queue length vector. 
In this scenario, the uncontrollable policy is a fixed mapping from the observed network event $\omega_t$ and the observed queue length vector $\mathbf{Q}(t)$ to a routing vector $\mathbf{f}^u(t)$ for uncontrollable nodes, i.e., $\pi_u:\Big(\omega_t,\mathbf{Q}(t)\Big)\mapsto \mathbf{f}^u(t)$. 

Note that the queueing dynamics are
\[
\begin{split}
\small
Q_{ik}(t+1)  & \le \Big[ Q_{ik}(t) + a_{ik}(t) + \sum_{j\in\mathcal{C}}f_{jik}(t) \\
&~~~~+ \sum_{j\in\mathcal{U}}f_{jik}(t)-\sum_{j\in\mathcal{N}}f_{ijk}(t)\Big]^+,~\forall i,k.
\end{split}
\]
Since for each $j\in\mathcal{U}$, its routing variable $f_{jik}(t)$ is an arbitrary (unknown) function of $\mathbf{Q}(t)$, the above queueing dynamics could depend on $\mathbf{Q}(t)$ in an arbitrary (unknown) way that is not in the simple piecewise-linear form as in the classic Lindley recursion. As a result, we rewrite the queueing dynamics as
\begin{equation}\label{eq:new_q}
\mathbf{Q}(t+1) = \beta(\mathbf{f}^c(t),\mathbf{Q}(t), \omega_t),
\end{equation}
where $\beta(\cdot)$ is some unknown function that depends on our controllable routing action $\mathbf{f}^c(t)$, the current queue length vector $\mathbf{Q}(t)$ and the observed network event $\omega_t$. 

Due to the unknown queueing dynamics, many analytical tools for optimal network control break down. For example, the previous Tracking-MaxWeight algorithm utilizes the Lyapunov drift analysis which is not applicable if the queueing dynamics do not follow  the Lindley recursion.  As a result, optimal network control becomes very challenging and fundamentally different from the traditional stochastic network optimization framework. In the following, we first formulate the problem a \textbf{Markov Decision Process (MDP) with unknown dynamics} and then propose a new reinforcement learning algorithm that can achieve network stability under mild conditions.

Before moving on to the technical details, we first introduce some notations and assumptions that will be used throughout this section. For convenience, we define action $\alpha_t\triangleq \mathbf{f}^c(t)$ and simply write ``controllable routing action $\mathbf{f}^c(t)$" as ``action $\alpha_t$", since the uncontrollable routing action $\mathbf{f}^u(t)$ has been implicitly treated as a part of the environment (see queueing dynamics \eqref{eq:new_q}).  For the same reason, ``controllable policy $\pi_c$" and ``policy $\pi$" are also used interchangeably. The action space for $\alpha_t$ is denoted by $\mathcal{A}$ which is assumed to be fixed and finite. We also make the following assumption regarding the optimal system performance.
\begin{assumption}\label{as:mdp-stability}
There exists a policy $\pi^*$ such that $\sum_{i,k} Q^*_{ik}(t)<\infty$ with probability 1 for any $t\ge 0$, where $\mathbf{Q}^*(t)$ is the queue length vector in slot $t$ under policy $\pi^*$.
\end{assumption}
\noindent In other words, it is required that the total queue length should remain bounded under an optimal policy $\pi^*$ otherwise there is no hope for stabilizing the network.  In essence, Assumption \ref{as:mdp-stability} requires that the network be stabilizable by some controllable policy $\pi^*$.

\subsection{MDP Formulation}
We formulate the problem of achieving network stability as an MDP $M=(\mathcal{A}, \mathcal{S}, \theta, P)$. Here $\mathcal{A}$ is the routing action space for controllable nodes, and $\mathcal{S}$ is the state space that corresponds to the queue length vector space $\mathcal{Q}$. The cost function $\theta(\alpha_t,\mathbf{Q}(t))$ under  action $\alpha_t$ and state $\mathbf{Q}(t)$ is given by $\theta(\alpha_t,\mathbf{Q}(t))=\sum_{i,k} Q_{ik}(t)$, which corresponds to the sum of queue lengths in  slot $t$. In addition, $P$ is the state transition matrix, where $P(\mathbf{Q}'|\mathbf{Q},\alpha)$ is the probability that the next state is $\mathbf{Q}'$ when  action $\alpha$ is taken under the current state $\mathbf{Q}$. Note that the transition matrix $P$ is generated according to the queueing dynamics \eqref{eq:new_q}, and that the influence of network event and uncontrollable routing action has been implicitly incorporated into the probabilistic transition matrix $P$. Note also that the queueing dynamics $\beta(\cdot)$ are unknown, so this is an MDP with unknown dynamics, which is also referred to as a \textbf{Reinforcement Learning (RL)} problem \cite{sutton}.

Let $J^\pi\Big(M,\mathbf{Q}(0)\Big)$ be the time-average expected total queue length  when  policy $\pi$ is applied in MDP $M$ and the initial queue length vector is $\mathbf{Q}(0)$, i.e.,
\[
J^{\pi}\Big(M,\mathbf{Q}(0)\Big) =\lim_{T\rightarrow \infty}\frac{1}{T}\sum_{t=0}^{T-1} \mathbb{E}_{\pi,M}\Big[\sum_{i,k} Q_{ik}(t)\Big|\mathbf{Q}(0)\Big],
\]
where the expectation $\mathbf{E}_{\pi,M}[\cdot]$ is with respect to the randomness of the queue length trajectory $\{\mathbf{Q}(t)\}_{t\ge 0}$ when policy $\pi$ is applied in MDP $M$. Also let $J^*\Big(M,\mathbf{Q}(0)\Big)=\min_\pi J^{\pi}\Big(M,\mathbf{Q}(0)\Big)$ be the minimum time-average expected queue length under an optimal policy $\pi^*$. Our objective is to find an optimal policy that solves the MDP and achieves the minimum average queue length.

\subsection{Challenges to Solving the MDP}
The MDP has an unknown transition structure, which gives rise to an ``exploration-exploitation" tradeoff. On one hand, we need to exploit the existing knowledge to make the best (myopic) decision; on the other hand, it is necessary to explore new states in order to learn which states may lead to lower costs in the future. Moreover, there might be some ``trapping"  sub-optimal states that  take  a long time (or is even impossible) for any policy to escape. Any algorithm that has zero knowledge about system dynamics at the beginning is likely to get trapped in these states during the exploration phase. Therefore, we need to impose restrictions on the transition structure in the MDP model. In particular, we restrict our consideration to weakly communicating MDPs with finite communication time, defined as follows.
\begin{assumption}\label{as:diameter}
For any two queue length vectors $\mathbf{Q}$ and $\mathbf{Q}'$ (except for those which are transient under every policy), there exists a policy $\pi$ that can move from $\mathbf{Q}$ to $\mathbf{Q}'$ within $L||\mathbf{Q}'-\mathbf{Q}||_1$ time slots (in expectation), where $L$ is a constant.
\end{assumption}
\noindent In other words, it is assumed that  there is no ``trapping" state in the system otherwise no reinforcement learning algorithm can be guranteed to avoid the traps and optimally solve the MDP. Note that in a weakly communicating MDP, the optimal average cost does not depend on the initial state (cf. \cite{MDP}, Section 8.3.3). Thus we drop the dependence on the initial state $\mathbf{Q}(0)$, and write the optimal average cost (queue length) as $J^*(M)$.

Another challenge is that the MDP has a countably-infinite state space (i.e., queue length vector space). Existing reinforcement learning methods that can handle such an infinite state space are mostly heuristic-based (e.g., \cite{dqn}\cite{ddpg}\cite{trpo}), and do not have any performance guarantees. On the other hand, there are a few reinforcement learning algorithms that do have good performance guarantees, but these algorithms require that the size of the state space be relatively small. Even if we consider a finite time horizon $T$, the size of the queue length vector space could be up to $O(T^N)$ (assuming bounded arrivals in each slot), which could lead to weak performance bounds. For example, in the UCRL algorithm \cite{UCRL1,UCRL2}, the regret bound is $O(S\sqrt{T})$, where $S$ is the size of the state space. If UCRL is applied in our context, the resulting regret bound would be $O(T^{N+0.5})$ which is a trivial super-linear regret bound.

\subsection{TUCRL Algorithm}
In this section, we develop an algorithm that achieves network stability under Assumptions \ref{as:mdp-stability} and \ref{as:diameter}. We call our algorithm Truncated Upper Confidence Reinforcement Learning (TUCRL), as it combines the model-based UCRL algorithm \cite{UCRL1,UCRL2} with a queue truncation technique that resolves the infinite state space problem.

Specifically, consider a \emph{truncated system} where new exogenous packet arrivals are dropped when the total queue length $\sum_{i,k} Q_{ik}(t)$ reaches $V-1$ for some threshold $V\ge 1$. In such a truncated system, the state space is the truncated queue length vector space $\mathcal{Q}_V$ which contains all queue length vectors where the length of each queue does not exceed $V-1$. In order for packet dropping to be feasible, we assume that there is an admission control action that can shed new exogenous packets as needed.

Our TUCRL algorithm applies the model-based UCRL algorithm \cite{UCRL1,UCRL2} in  the truncated system, which maintains an estimation for the unknown queueing dynamics  and then computes the optimal policy under the estimated dynamics.  It applies the ``optimistic principle"  for exploration, where under-explored state-action pairs are assumed to be able to result in lower costs, which implicitly encourages the exploration of novel state-action pairs.

The detailed description of TUCRL is presented in Algorithm \ref{alg:RL}, which is similar to the standard UCRL algorithm  except that queue truncation is applied when appropriate. Specifically, the TUCRL algorithm proceeds in episodes, and the length of each episode is dynamically determined. In episode $\ell$, the TUCRL algorithm first constructs an empirical estimation $\hat{P}$ for the transition matrix based on historical observations (step 1). In particular, the estimated transition probability from state $\mathbf{Q}$ to $\mathbf{Q}'$ under action $\alpha$ is 
\begin{equation}\label{eq:transition-estimation}
\hat{P}(\mathbf{Q}'|\mathbf{Q},\alpha)=\frac{n_\ell(\mathbf{Q},\alpha,\mathbf{Q}')}{n_\ell(\mathbf{Q},\alpha)},
\end{equation}
where $n_\ell(\mathbf{Q},\alpha)$ is the cumulative number of visits to state-action pair $(\mathbf{Q},\alpha)$ up until the beginning of episode $\ell$ and $n_{\ell}(\mathbf{Q},\alpha,\mathbf{Q}')$ is the number of times that transition $(\mathbf{Q},\alpha)\rightarrow \mathbf{Q}'$ happens up to the beginning of episode $\ell$. Note that if $n_{\ell}(\mathbf{Q},\alpha)=0$, the estimated transition probability is set to be zero.

Then the TUCRL algorithm constructs an upper confidence set $\mathcal{M}_{\ell}$ for all plausible MDP models based on the empirical estimation $\hat{P}$ (step 2). The upper confidence set is constructed in a way such that it contains the true MDP model with high probability. Specifically, the upper confidence set $\mathcal{M}_{\ell}$ contains all the MDPs with truncated queue length space $\mathcal{Q}_V$ and transition matrix $P\in\mathcal{P}_{\ell}$ where
\begin{equation}\label{eq:confidence-set}
\small
\mathcal{P}_{\ell}=\Big\{P:\Big|\Big|P(\cdot|\mathbf{Q},\alpha)-\hat{P}(\cdot|\mathbf{Q},\alpha)\Big|\Big|_1 \le \sqrt{\frac{ C\log(2|\mathcal{A}|t_{\ell} V)}{\max\{1,n_{\ell}(\mathbf{Q},\alpha)\}}} \Big\}.
\end{equation}
Here, $V$ is the queue truncation threshold, $t_{\ell}$ is the starting time of episode $\ell$ and $C\triangleq 2(2ND+1)^N$ is a constant.

Next, the TUCRL algorithm selects an ``optimistic MDP" $M_{\ell}$ that yields the minimum average queue length among all the plausible MDPs in the confidence set $\mathcal{M}_{\ell}$, and  computes a nearly-optimal policy $\pi_{\ell}$ under MDP $M_{\ell}$ (step 3). The joint selection of the optimistic MDP and the calculation of the nearly-optimal policy are referred to an \textbf{optimistic planning} \cite{optimistic}. There are many efficient methods for performing optimistic planning, such as Extended Value Iteration \cite{UCRL2} and OP-MDP \cite{op-mdp}. For completeness, we provide a description of Extended Value Iteration in Appendix \ref{ap:EVI}. 

Finally, the computed policy $\pi_{\ell}$ is executed until the stopping condition of episode $\ell$ is triggered. An episode ends when the number of visits to some state-action pair doubles, i.e., when we encounter a state-action pair $(\mathbf{Q}(t),\alpha_t)$ such that its visiting frequency in episode $\ell$ ($v_{\ell}(\mathbf{Q}(t),\alpha_t)$) equals its cumulative visiting frequency up to the beginning of episode $\ell$ ($n_{\ell}(\mathbf{Q}(t),\alpha_t)$). We will show that this  stopping condition guarantees that the  total number of episodes up to time $T$ is $O(V^N \log T)$ (see Appendix \ref{as:epi}). Note that during the execution of policy $\pi_{\ell}$, new packet arrivals may be dropped if the total queue length exceeds $V-1$. Here, the dropped packets could be any new arrivals to any queue. We will prove that the fraction of dropped packets is negligible if the threshold $V$ is properly selected.

\begin{algorithm}
\caption{Truncated Upper Confidence Reinforcement Learning (TUCRL)}\label{alg:RL}
\begin{algorithmic}
\REQUIRE  queue truncation threshold $V$ \\

\vspace{1mm}
\STATE Set $t=0$
\FOR{episode $\ell=1,2,\cdots$} \vspace{2mm}
\STATE \textbf{1. Initialize episode $\ell$:}\label{step:init}
\STATE \quad $\bullet$ Set the start of episode $\ell$: $t_{\ell}=t$
\STATE \quad $\bullet$ Initialize state-action count for episode $\ell$: $v_{\ell}(\mathbf{Q},\alpha)=0$ 
\STATE \quad $\bullet$ Update accumulative state-action count $n_{\ell}(\mathbf{Q},\alpha)$ and transition count $n_{\ell}(\mathbf{Q},\alpha,\mathbf{\mathbf{Q}}')$ up to episode $\ell$
\STATE \quad $\bullet$ Estimate transition probability $\hat{P}(\mathbf{Q}'|\mathbf{Q},\alpha)$ according to \eqref{eq:transition-estimation} for any $\mathbf{Q},\mathbf{Q}'\in\mathcal{Q}_V$ and $\alpha\in\mathcal{A}$  \vspace{3mm}
\STATE \textbf{2. Construct upper confidence set:}\label{step:construct}
\STATE \quad  Construct a confidence set $\mathcal{M}_{\ell}$ that contains all the MDPs with truncated queue space $\mathcal{Q}_V$ and transition matrix $P\in\mathcal{P}_{\ell}$ as shown in \eqref{eq:confidence-set}  \vspace{3mm}
\STATE \textbf{3. Optimistic planning:}\label{step:plan}
\STATE \quad Compute the optimistic MDP model $M_{\ell}$ (that yields the minimum average total queue length) in the confidence set $\mathcal{M}_{\ell}$ and a nearly-optimal policy $\pi_{\ell}$ under $M_{\ell}$ (up to accuracy $\frac{1}{\sqrt{t_{\ell}}}$) \vspace{3mm}
\STATE \textbf{4. Execute policy (with packet dropping):}\label{step:ex}
\REPEAT
\STATE $\bullet$ Observe current queue length vector $\mathbf{Q}(t)$ and new exogenous arrivals $\mathbf{a}(t)$
\STATE $\bullet$ Arbitrarily drop $\Big[\sum_{i,k} \Big(Q_{ik}(t) + a_{ik}(t)\Big) - V+1\Big]^+$ newly arrived packets from the network
\STATE $\bullet$ Take action $\alpha_t=\pi_{\ell}(\mathbf{Q}(t))$ 
\STATE $\bullet$ Update $v_{\ell}(\mathbf{Q}(t),\alpha_t)=v_{\ell}(\mathbf{Q}(t),\alpha_t) + 1$
\STATE $\bullet$ $t=t+1$
\UNTIL{$v_{\ell}(\mathbf{Q}(t),\alpha_t)=\max\{1,n_{\ell}(\mathbf{Q}(t),\alpha_t)\}$}
\label{step:stopping}
\ENDFOR
\end{algorithmic}
\end{algorithm}

\subsection{Performance of TUCRL Algorithm}
The following theorem characterizes the performance of the TUCRL algorithm regarding its queue length, packet dropping rate and convergence rate.
\begin{theorem}\label{thm:tucrl-final}
Under Assumptions \ref{as:mdp-stability} and \ref{as:diameter}, the performance of the TUCRL algorithm is as follows.

\vspace{1mm}

$\bullet$ \textbf{(Queue Length)} The time-average expected queue length converges to a bounded value:
\[
\lim_{T\rightarrow\infty}\frac{1}{T}\sum_{t=0}^{T-1}\sum_{i,k} \mathbb{E}[Q_{ik}(t)] \le \Theta(1)-\Theta\Big(\frac{1}{V}\Big).
\]

$\bullet$ \textbf{(Packet Dropping Rate)} The long-term expected fraction of dropped packets is
\[
\lim_{T\rightarrow\infty} \mathbb{E}[\eta_T] \le \Theta\Big(\frac{1}{V}\Big),
\]
where $\eta_T$ is the fraction of dropped packets within $T$ slots.

$\bullet$ \textbf{(Convergence Rate)} The time-average expected queue length after $T$ slots is within a
$\widetilde{\mathcal{O}}\Big(\frac{poly(V)}{\sqrt{T}}\Big)$-neighborhood of the steady-state expected queue length,
where $poly(V)$ is some polynomial in $V$ and $\widetilde{\mathcal{O}}$ is the big-O notation that ignores any logarithmic term.
\end{theorem}
\begin{proof}
We first find an upper bound on the total queue length under TUCRL in the truncated system. Based on the queue length upper bound, we further analyze the fraction of time when queue truncation is triggered by using concentration inequalities. See Appendix \ref{ap:tucrl-final} for details.
\end{proof}

\vspace{1mm}

There are several important observations regarding Theorem \ref{thm:tucrl-final}. First, the TUCRL algorithm achieves bounded queue length by dropping a negligible fraction of packets under a suitably large value of $V$. 
Second, there is a \textbf{three-way tradeoff} between total queue length (delay), packet dropping rate (throughput) and convergence rate. For example, by increasing the value of $V$, the packet dropping rate becomes smaller (i.e., throughput becomes higher) but the total queue length (delay) increases and the convergence becomes slower. Similar three-way tradeoffs between utility, delay and convergence rate are discussed in \cite{ness-heavy}.

\vspace{2mm}

\noindent \textbf{Complexity of TUCRL.} The time complexity of TUCRL is dominated by the complexity of the optimistic planning module (step 3) which is implementation-dependent. For example, if a naive Extended Value Iteration (see Appendix \ref{ap:EVI}) is used, the time complexity of each value iteration step  is exponential in the number of queues and thus cannot scale to large-scale problems. One way to scale the optimisic planning module is by using approximate dynamic programming that employs various approximation techniques in the planning procedure, such as using linear functions or neural networks to approximate the value function (see \cite{ADP} for a comprehensive introduction). Recent deep reinforcement learning techniques may also be leveraged to efficiently perform value iterations in large-scale problems, such as Randomized Least-Squares Value Iteration (RLSVI) \cite{deep-exploration}, Value Iteration Networks (VIN) \cite{VIN} and Value Prediction Networks (VPN) \cite{VPN}.
Such approximations will not lead to significant changes in the performance of TUCRL since we only require an approximate solution in step 3.

\section{Simulation Results}\label{sec:simulation}
\subsection{Scenario 1: Queue-Agnostic Uncontrollable Policy}
We first study the partially-controllable network shown in Figure \ref{fig:example_2}. There are two flows: $1\rightarrow 4$ and $6\rightarrow 4$.  Each node in the network needs to make a routing and scheduling decision in every time slot. The constraint is that each node can transmit to only one of its neighbors in each time slot and the transmission rate  over each link cannot exceed its capacity. Node 2 and node 3 are uncontrollable nodes that use randomized queue-agnostic policies. Specifically, uncontrollable node 2 uses a randomized routing algorithm that transmits any packets it received to either node 3 or node 5 with an equal probability in each time slot. Uncontrollable node 3 uses a randomized scheduling policy that serves flow $1\rightarrow 4$ or flow $6\rightarrow 4$ with an equal probability in each time slot. The arrival rate of flow $6\rightarrow 4$ is 5. In this case, it can be shown that the maximum supportable arrival rate for flow $1\rightarrow 4$ is 25 given the routing constraints and the behavior of uncontrollable nodes.  

\begin{figure}[ht!]
\begin{center}
\includegraphics[width=2.0in]{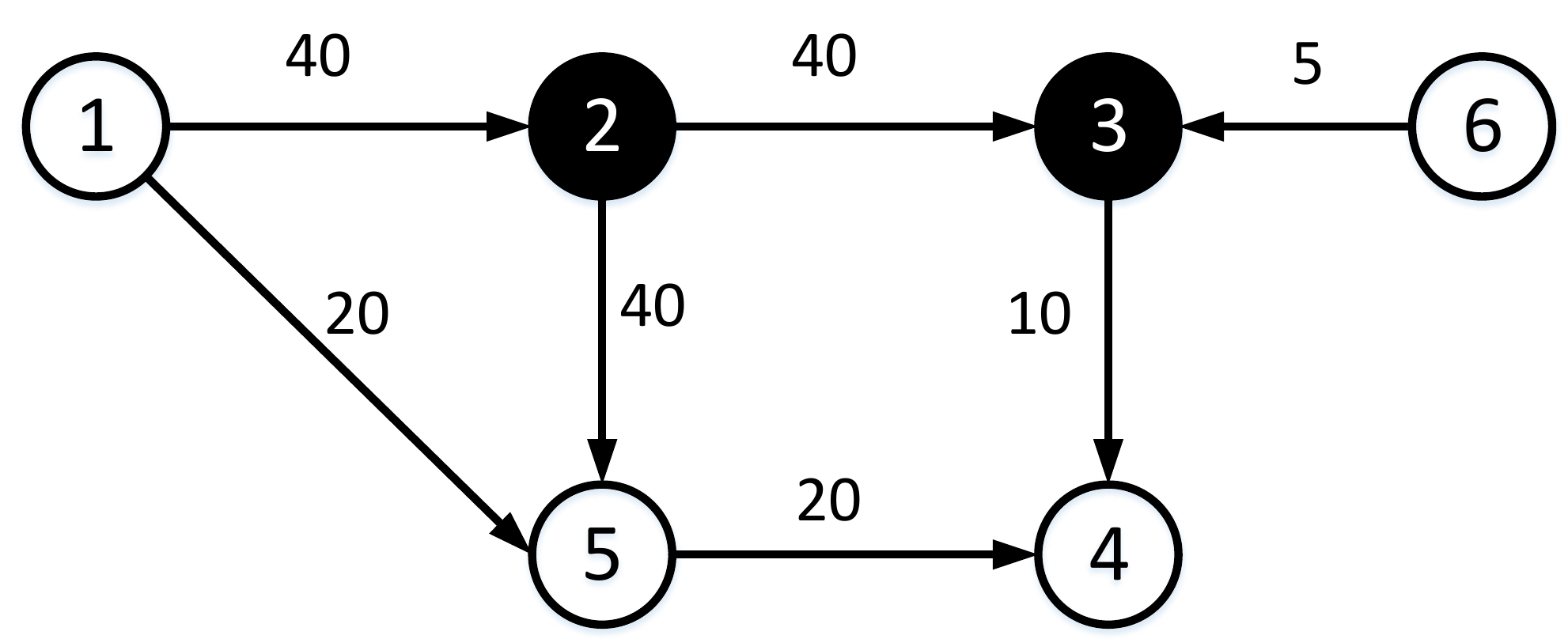}
\caption{Network topology used in simulation scenario 1. The number next to each link is its capacity. Each node can only transmit to one of its neighbors in each slot. Black nodes are uncontrollable nodes that use randomized queue-agnostic policies.}
\label{fig:example_2}
\end{center}
\end{figure}

\begin{figure*}[ht!]
\subfigure[Throughput performance of MaxWeight and Tracking-MaxWeight.]
{\label{fig:tmw_load}
\includegraphics[width=55mm,height=45mm]{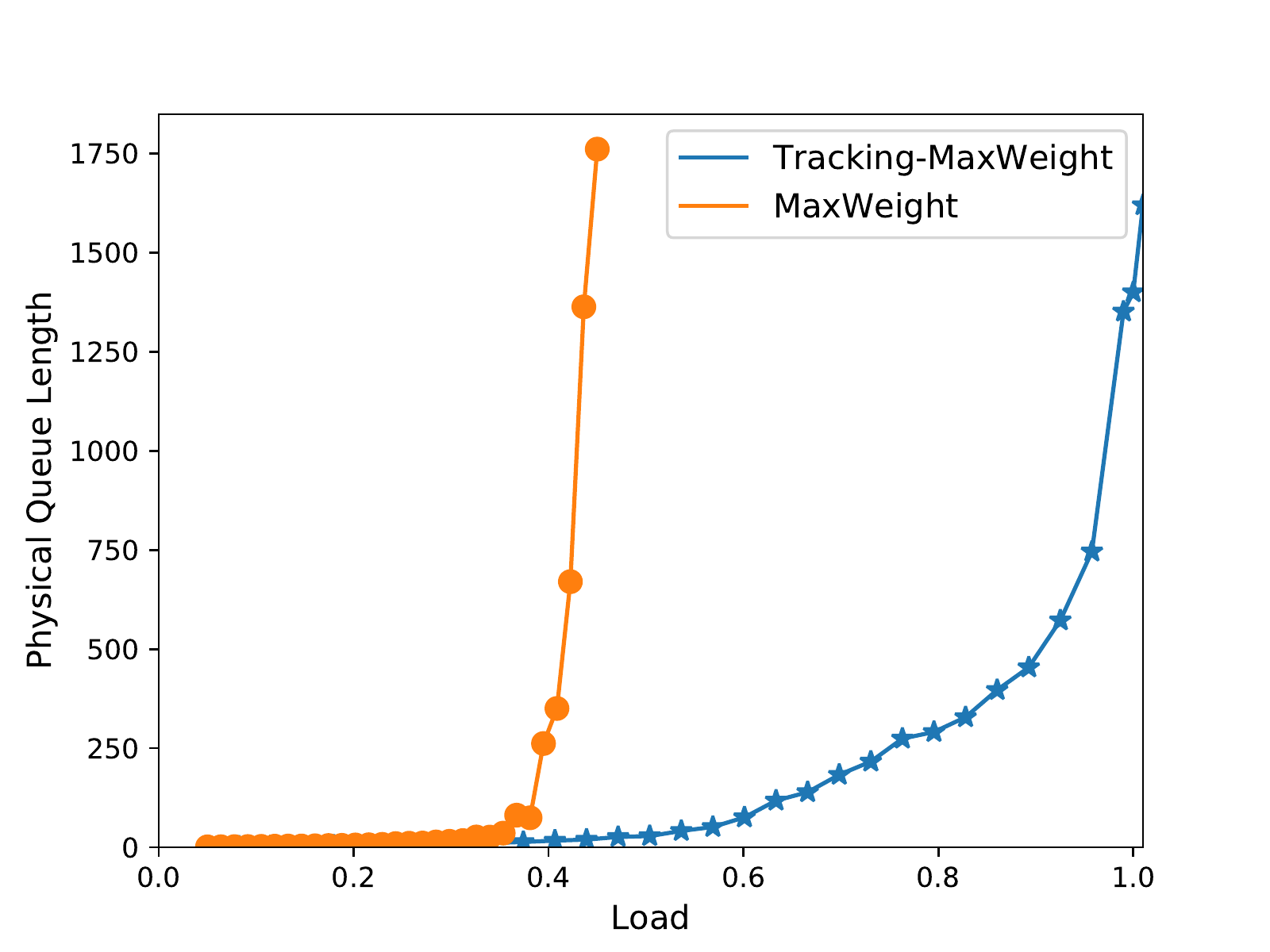}}\hspace{3mm}
\subfigure[Queue length under the TMW algorithm (load = 0.99).]
{\label{fig:tmw_queue}\includegraphics[width=55mm,height=42mm]{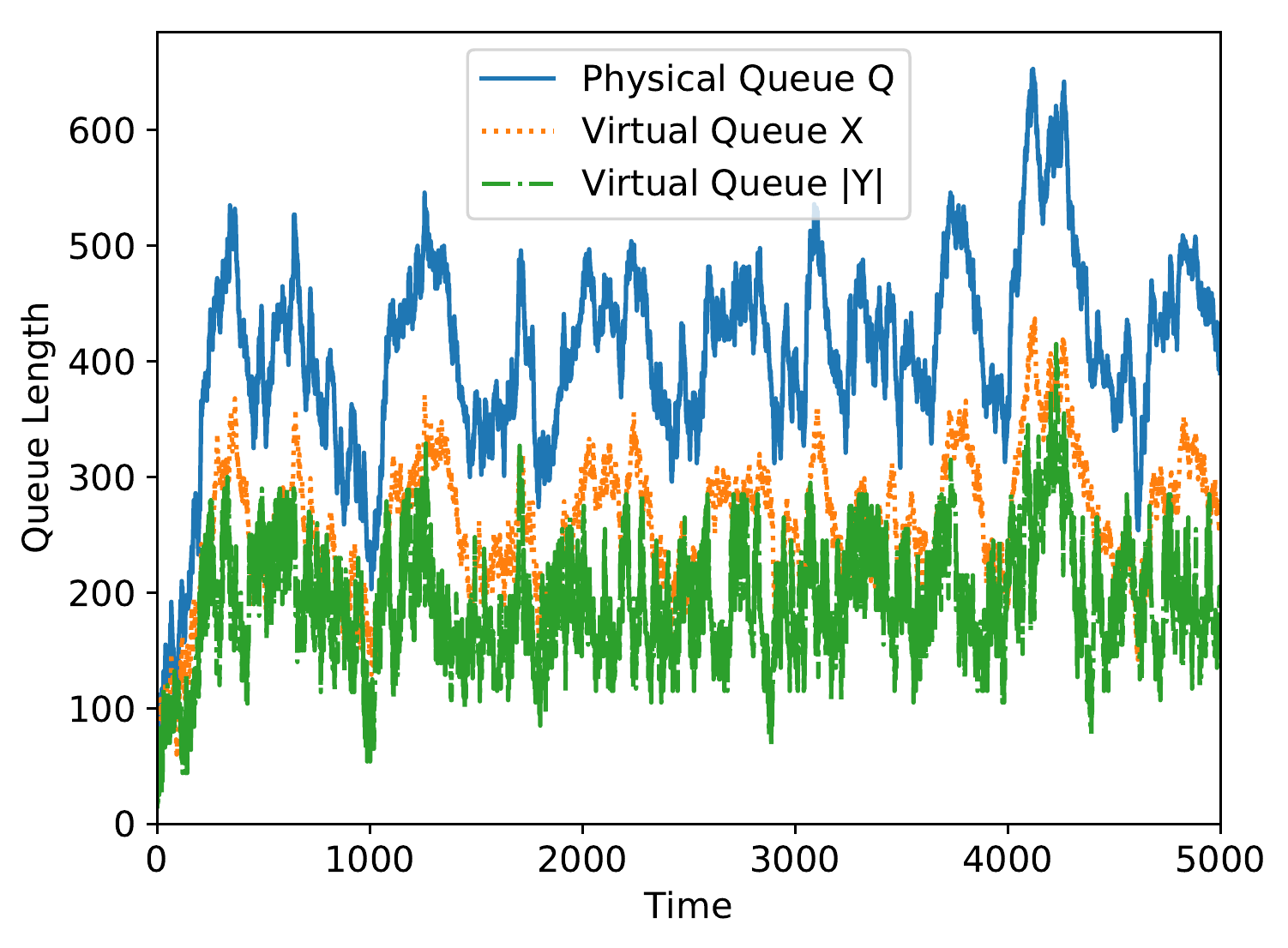}}\hspace{3mm}
\subfigure[The TMW algorithm quickly learns that node 3 serves flow $1\rightarrow 4$ with probability 0.5.]
{\label{fig:tmw_estimation}\includegraphics[width=55mm,height=42mm]{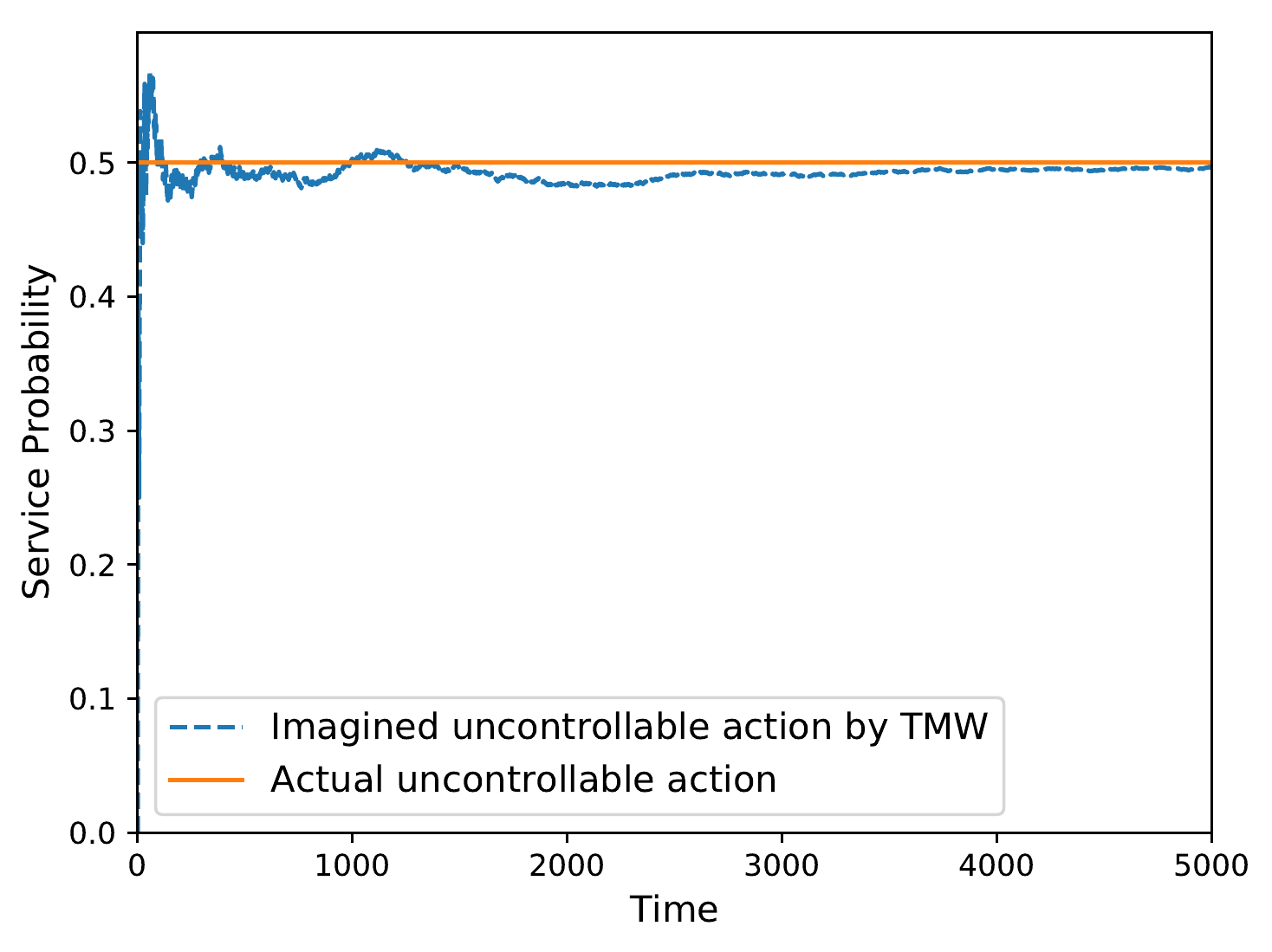}}
\caption{Performance of the Tracking-MaxWeight (TMW) algorithm in Scenario 1.}
\label{fig:scenario1}
\end{figure*}

We have shown in Section \ref{sec:agnostic} that the Tracking-MaxWeight (TMW) algorithm achieves the optimal throughput in this scenario. In Figure \ref{fig:tmw_load}, we compare Tracking-MaxWeight  with the well-known MaxWeight algorithm (i.e., BackPressure routing), in terms of the supportable rate for flow $1\rightarrow 4$. Specifically, Figure \ref{fig:tmw_load} shows the total queue length achieved by MaxWeight and Tracking-MaxWeight under different system loads (if the load is $\rho$, then the arrival rate of flow $1\rightarrow 4$ is $25\rho$ while the arrival rate of flow $6\rightarrow 4$ is fixed to 5). It is observed that MaxWeight can only support around 40\% arrivals (the queue length under MaxWeight blows up at load $\approx 0.4$). By comparison, our Tracking-MaxWeight achieves the optimal throughput.

We further examine the behavior of the Tracking-MaxWeight algorithm in Figure \ref{fig:tmw_queue} and Figure \ref{fig:tmw_estimation}. Specifically, Figure \ref{fig:tmw_queue} shows the queue length trajectory for the physical queue  $\mathbf{Q}(t)$ and the two virtual queues $\mathbf{X}(t),\mathbf{Y}(t)$. As our theory predicts, both the physical queue $\mathbf{Q}(t)$ and the two virtual queues $\mathbf{X}(t),\mathbf{Y}(t)$ are stable under the TMW algorithm. Figure \ref{fig:tmw_estimation} shows the learning curve of the TMW algorithm for the uncontrollable policy used by node $3$. In particular, node $3$ uses  randomized scheduling that serves flow $1\rightarrow 4$ and flow $6\rightarrow 4$ with an equal probability 0.5. It is observed in Figure \ref{fig:tmw_estimation} that the TMW algorithm quickly learns the service probability for flow $1\rightarrow 4$ at node 3 (i.e., the ``imagined uncontrollable action" in TMW approaches the true uncontrollable action).

\subsection{Scenario 2: Queue-Dependent Uncontrollable Policy}
Next we study a more challenging scenario where the action taken by uncontrollable nodes is queue-dependent. In particular, consider the network topology shown in Figure \ref{fig:sim_topology} where node 2 and node 3 are uncontrollable. There is only one flow $1\rightarrow 4$ and the constraint is that each node can transmit to only one of its neighbours in each time slot. The policy used by the two uncontrollable nodes is as follows. Let $\mu_{24}(t)$ and $\mu_{34}(t)$ be the transmission rate that node 2 and node 3 allocates to the flow in slot $t$, respectively. Then
\[
\small
(\mu_{24}(t),\mu_{34}(t))=
\begin{cases} 
      (0.5, 0) &  Q_3(t)\le 10 \\
      (0, 1) & Q_2(t)\le 10\text{ and }Q_3(t)>10 \\
      (0.25, 0.25) & Q_2(t)> 10\text{ and }Q_3(t)>10
  \end{cases}
\]
As a result, the maximum throughput of 1 can be supported only if $Q_2(t)$ is small ($Q_2(t)\le 10$) and $Q_3(t)$ is large ($Q_3(t)>10$). Although this is an artificial example, it sheds light on the challenges when uncontrollable nodes use queue-dependent policies: any throughput-optimal algorithm should be able to efficiently learn which queue length region can support the maximum throughput and keep the queue length within this region.  

\begin{figure}[ht!]
\begin{center}
\includegraphics[width=1.0in]{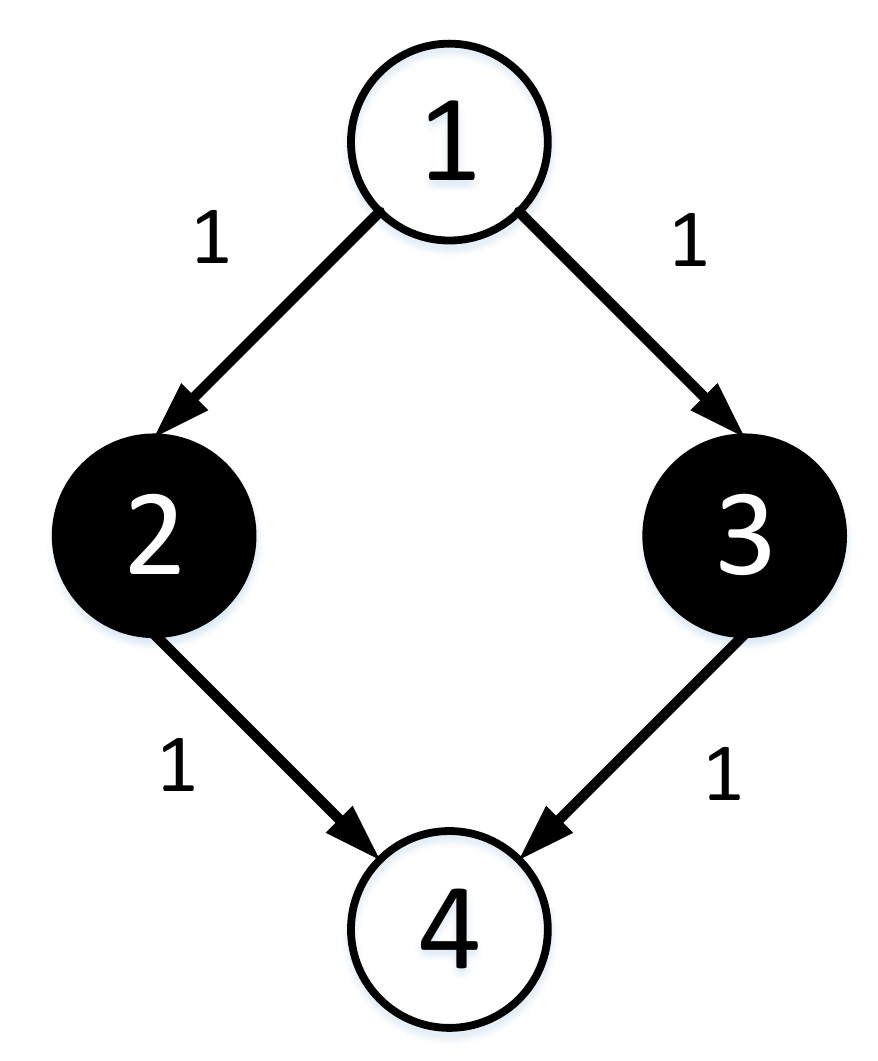}
\caption{Network topology used in simulation scenario 2 The number next to each link is its capacity. Each node can only transmit to one of its neighbors in each slot. There is only one flow  $1\rightarrow 4$. Black nodes are uncontrollable nodes that use queue-dependent policies.}
\label{fig:sim_topology}
\end{center}
\end{figure}

\begin{figure*}[htbp]
\centering
\begin{minipage}[t]{0.32\textwidth}
\centering
\includegraphics[width=2.2in, height=1.9in]{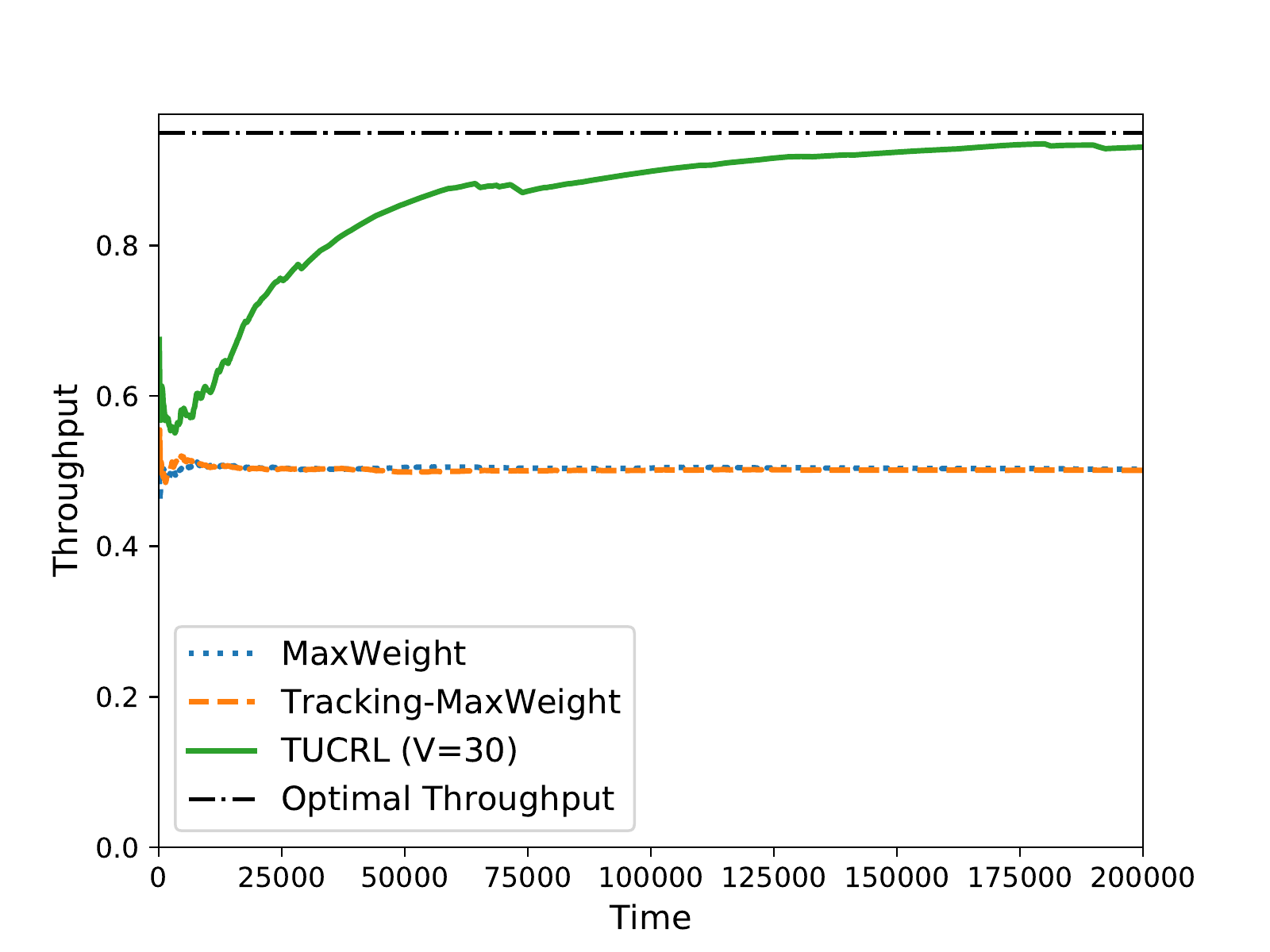}\vspace{-3mm}
\caption{Throughput comparison among MaxWeight, Tracking MaxWeight and TUCRL in Scenario 2 (load = 0.95).}
\label{fig:rl_comp}
\end{minipage}%
\hspace{0.16cm}
\begin{minipage}[t]{0.32\textwidth}
\centering
\includegraphics[width=2.2in, height=1.9in]{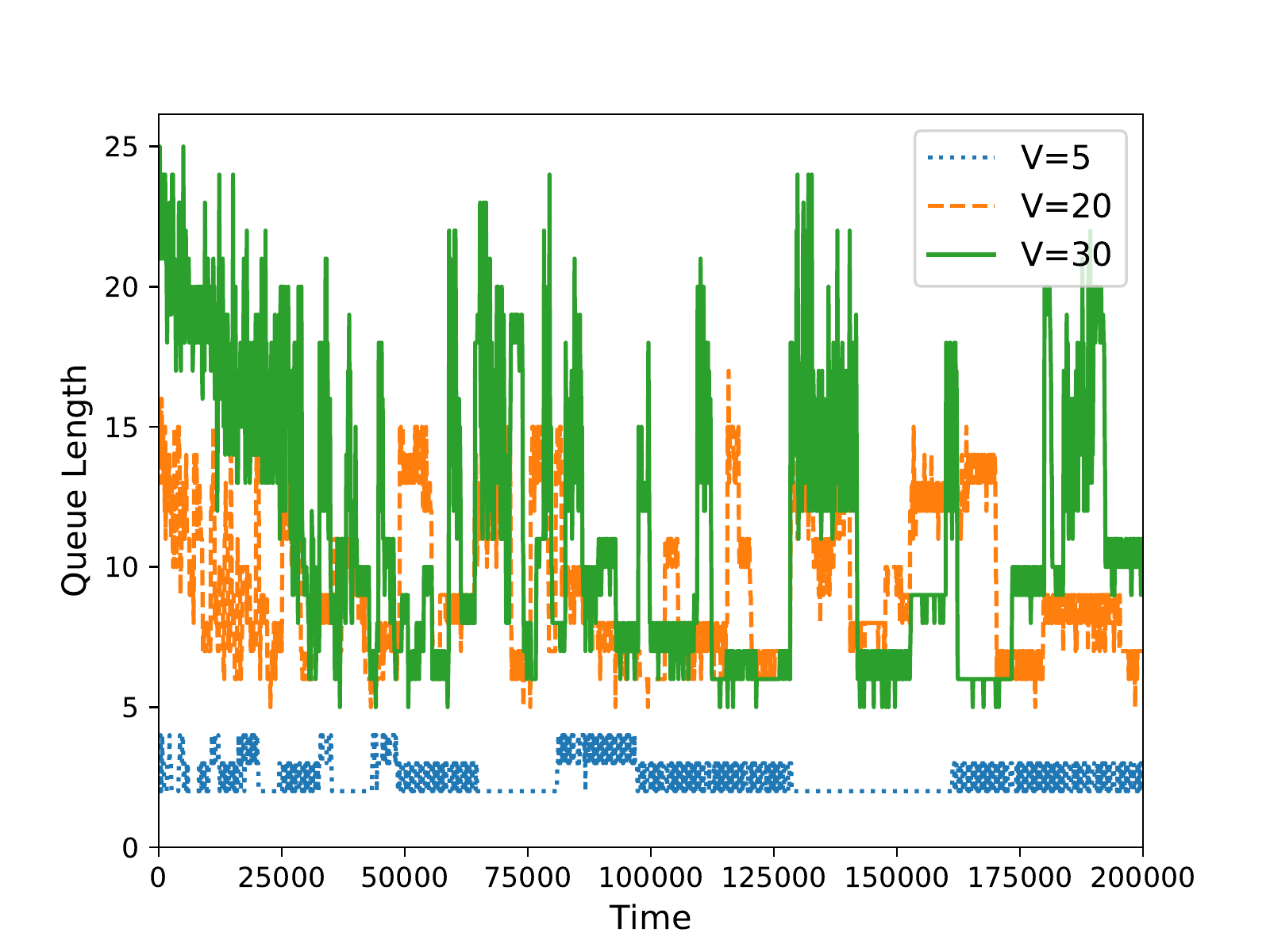}\vspace{-3mm}
\caption{Queue length under the TUCRL algorithm with different queue truncation threshold $V$ (load = 0.95).}
\label{fig:rl_queue}
\end{minipage}
\hspace{0.16cm}
\begin{minipage}[t]{0.32\textwidth}
\centering
\includegraphics[width=2.2in, height=1.9in]{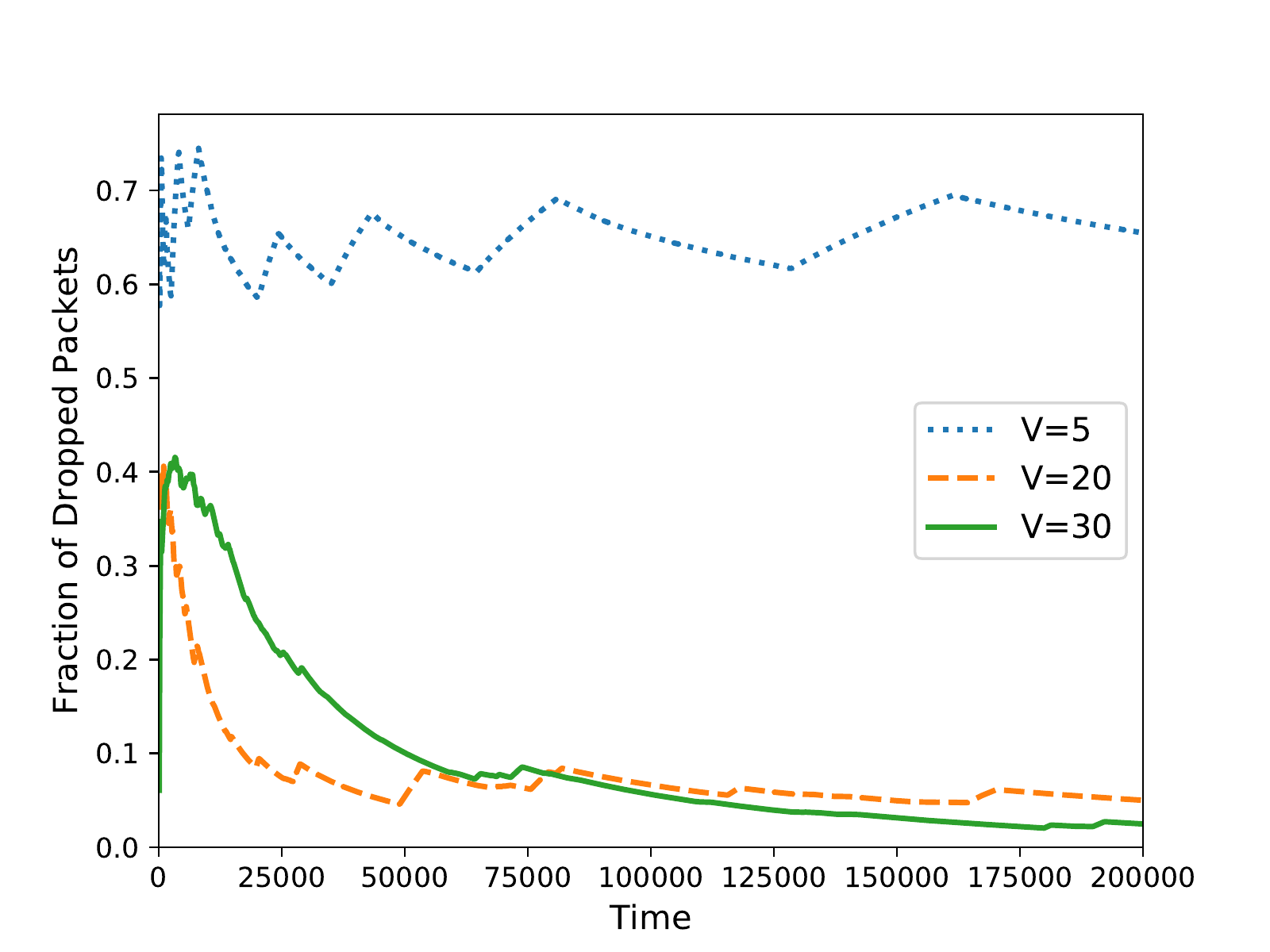}\vspace{-3mm}
\caption{Fraction of dropped packets under the TUCRL algorithm with different queue truncation thresholds $V$ (load = 0.95).}
\label{fig:rl_loss}
\end{minipage}
\end{figure*}

We first compare the throughput performance of TUCRL with MaxWeight and Tracking-MaxWeight. Note that the TUCRL algorithm occasionally drops packets. In order to make a fair comparison, the throughput performance is measured with respect to the number of packets that have been delivered. It is observed in Figure \ref{fig:rl_comp} that TUCRL achieves the optimal throughput while MaxWeight or Tracking-MaxWeight only deliver a throughput of 0.5 in this scenario. It should be noted that the TUCRL algorithm takes longer time to learn and converge than MaxWeight or Tracking-MaxWeight.


Next we investigate the performance of the TUCRL algorithm under different values of the truncation threshold $V$. As we proved in Theorem \ref{thm:tucrl-final}, the value of $V$ determines a three-way tradeoff between queue length, packet dropping rate and convergence rate. As is illustrated in Figure \ref{fig:rl_queue} and Figure \ref{fig:rl_loss}, a larger value of $V$ leads to a larger queue length and the convergence becomes slower, but the fraction of dropped packets becomes smaller. Note that when $V=20$ and $V=30$, the fraction of dropped packets becomes very small as time goes by. In contrast, when $V=5$, the fraction of dropped packets remains non-negligible ($\sim$ 60\%) since the TUCRL algorithm cannot explore the ``throughput-optimal region" where $Q_2(t)\le 10\text{ and }Q_3(t)>10$ with a queue truncation threshold $V=5$.


\section{Conclusions}\label{sec:conclusion}
In this paper, we study optimal network control algorithms that stabilize a partially-controllable network where a subset of nodes are uncontrollable. We first study the scenario where the uncontrollable nodes use a queue-agnostic policy and propose a simple throughput-optimal Tracking-MaxWeight algorithm that enhances the original MaxWeight algorithm with an explicit learning of uncontrollable behavior. Then we investigate the scenario where the uncontrollable policy may be queue-dependent. This problem is formulated as an MDP and we develop a reinforcement learning algorithm called TUCRL that achieves a three-way tradeoff between throughput, delay and convergence rate.

\begin{footnotesize}
\bibliographystyle{plain}
\bibliography{citations}

\begin{thebibliography}{10}

\bibitem{UCRL1}
Peter Auer and Ronald Ortner.
\newblock Logarithmic online regret bounds for undiscounted reinforcement
  learning.
\newblock In {\em Advances in Neural Information Processing Systems}, pages
  49--56, 2007.

\bibitem{UCRL2}
Thomas Jaksch, Ronald Ortner, and Peter Auer.
\newblock Near-optimal regret bounds for reinforcement learning.
\newblock {\em Journal of Machine Learning Research}, 11(Apr):1563--1600, 2010.

\bibitem{jones-overlay}
Nathaniel~M Jones, Georgios~S Paschos, Brooke Shrader, and Eytan Modiano.
\newblock An overlay architecture for throughput optimal multipath routing.
\newblock {\em IEEE/ACM Transactions on Networking}, 2017.

\bibitem{ac}
Vijay~R Konda and John~N Tsitsiklis.
\newblock Actor-critic algorithms.
\newblock In {\em Advances in neural information processing systems}, pages
  1008--1014, 2000.

\bibitem{ddpg}
Timothy~P Lillicrap, Jonathan~J Hunt, Alexander Pritzel, Nicolas Heess, Tom
  Erez, Yuval Tassa, David Silver, and Daan Wierstra.
\newblock Continuous control with deep reinforcement learning.
\newblock {\em arXiv preprint arXiv:1509.02971}, 2015.

\bibitem{lindley}
David~V Lindley.
\newblock The theory of queues with a single server.
\newblock In {\em Mathematical Proceedings of the Cambridge Philosophical
  Society}, volume~48, pages 277--289. Cambridge University Press, 1952.

\bibitem{ness-heavy}
Jia Liu, Atilla Eryilmaz, Ness~B Shroff, and Elizabeth~S Bentley.
\newblock Heavy-ball: A new approach to tame delay and convergence in wireless
  network optimization.
\newblock In {\em INFOCOM 2016-The 35th Annual IEEE International Conference on
  Computer Communications, IEEE}, pages 1--9. IEEE, 2016.

\bibitem{dqn}
Volodymyr Mnih, Koray Kavukcuoglu, David Silver, Alex Graves, Ioannis
  Antonoglou, Daan Wierstra, and Martin Riedmiller.
\newblock Playing atari with deep reinforcement learning.
\newblock {\em arXiv preprint arXiv:1312.5602}, 2013.

\bibitem{optimistic}
R{\'e}mi Munos et~al.
\newblock From bandits to monte-carlo tree search: The optimistic principle
  applied to optimization and planning.
\newblock {\em Foundations and Trends{\textregistered} in Machine Learning},
  7(1):1--129, 2014.

\bibitem{neely-stability}
Michael~J Neely.
\newblock Stability and capacity regions or discrete time queueing networks.
\newblock {\em arXiv preprint arXiv:1003.3396}, 2010.

\bibitem{neely-sno}
Michael~J Neely.
\newblock Stochastic network optimization with application to communication and
  queueing systems.
\newblock {\em Synthesis Lectures on Communication Networks}, 3(1):1--211,
  2010.

\bibitem{VPN}
Junhyuk Oh, Satinder Singh, and Honglak Lee.
\newblock Value prediction network.
\newblock In {\em Advances in Neural Information Processing Systems}, pages
  6120--6130, 2017.

\bibitem{PSRL2}
Ian Osband, Daniel Russo, and Benjamin Van~Roy.
\newblock (more) efficient reinforcement learning via posterior sampling.
\newblock In {\em Advances in Neural Information Processing Systems}, pages
  3003--3011, 2013.

\bibitem{deep-exploration}
Ian Osband, Daniel Russo, Zheng Wen, and Benjamin Van~Roy.
\newblock Deep exploration via randomized value functions.
\newblock {\em arXiv preprint arXiv:1703.07608}, 2017.

\bibitem{PSRL1}
Yi~Ouyang, Mukul Gagrani, Ashutosh Nayyar, and Rahul Jain.
\newblock Learning unknown markov decision processes: A thompson sampling
  approach.
\newblock In {\em Advances in Neural Information Processing Systems}, pages
  1333--1342, 2017.

\bibitem{paschos-overlay}
Georgios~S Paschos and Eytan Modiano.
\newblock Throughput optimal routing in overlay networks.
\newblock In {\em Communication, Control, and Computing (Allerton), 2014 52nd
  Annual Allerton Conference on}, pages 401--408. IEEE, 2014.

\bibitem{ADP}
Warren~B Powell.
\newblock {\em Approximate Dynamic Programming: Solving the curses of
  dimensionality}, volume 703.
\newblock John Wiley \& Sons, 2007.

\bibitem{MDP}
Martin~L Puterman.
\newblock {\em Markov decision processes: discrete stochastic dynamic
  programming}.
\newblock John Wiley \& Sons, 2014.

\bibitem{rai-overlay}
Anurag Rai, Rahul Singh, and Eytan Modiano.
\newblock A distributed algorithm for throughput optimal routing in overlay
  networks.
\newblock {\em arXiv preprint arXiv:1612.05537}, 2016.

\bibitem{trpo}
John Schulman, Sergey Levine, Pieter Abbeel, Michael Jordan, and Philipp
  Moritz.
\newblock Trust region policy optimization.
\newblock In {\em International Conference on Machine Learning}, pages
  1889--1897, 2015.

\bibitem{overlay-4}
Ramesh~K Sitaraman, Mangesh Kasbekar, Woody Lichtenstein, and Manish Jain.
\newblock Overlay networks: An akamai perspective.
\newblock {\em Advanced Content Delivery, Streaming, and Cloud Services},
  51(4):305--328, 2014.

\bibitem{rl-book}
Richard~S Sutton and Andrew~G Barto.
\newblock {\em Reinforcement learning: An introduction}, volume~1.
\newblock MIT press Cambridge, 1998.

\bibitem{sutton}
Richard~S Sutton and Andrew~G Barto.
\newblock {\em Reinforcement learning: An introduction}, volume~1.
\newblock MIT press Cambridge, 1998.

\bibitem{op-mdp}
Bal{\'a}zs Sz{\"o}r{\'e}nyi, Gunnar Kedenburg, and Remi Munos.
\newblock Optimistic planning in markov decision processes using a generative
  model.
\newblock In {\em Advances in Neural Information Processing Systems}, pages
  1035--1043, 2014.

\bibitem{VIN}
Aviv Tamar, Yi~Wu, Garrett Thomas, Sergey Levine, and Pieter Abbeel.
\newblock Value iteration networks.
\newblock In {\em Advances in Neural Information Processing Systems}, pages
  2154--2162, 2016.

\bibitem{tassiulas}
Leandros Tassiulas and Anthony Ephremides.
\newblock Stability properties of constrained queueing systems and scheduling
  policies for maximum throughput in multihop radio networks.
\newblock {\em IEEE transactions on automatic control}, 37(12):1936--1948,
  1992.

\bibitem{l1-inequality}
Tsachy Weissman, Erik Ordentlich, Gadiel Seroussi, Sergio Verdu, and Marcelo~J
  Weinberger.
\newblock Inequalities for the l1 deviation of the empirical distribution.
\newblock {\em Hewlett-Packard Labs, Tech. Rep}, 2003.

\bibitem{reinforce}
Ronald~J Williams.
\newblock Simple statistical gradient-following algorithms for connectionist
  reinforcement learning.
\newblock In {\em Reinforcement Learning}, pages 5--32. Springer, 1992.

\end{thebibliography}
\end{footnotesize}

\onecolumn

\appendices

\section{Proof to Theorem \ref{thm:TMW}}\label{ap:TMW}
We first introduce a lemma that  characterizes the load of the two virtual queues $\mathbf{X}(t)$ and $\mathbf{Y}(t)$. 
\begin{lemma}\label{lm:loaded-stochastic}
There exists an $\omega$-only controllable policy $\widehat{\pi}_c: \omega_t\mapsto\bm{\mu}^c(t)$ and an $\omega$-only uncontrollable policy $\widehat{\pi}_u: \omega_t\mapsto \bm{\mu}^u(t)$ such that
\[
\lambda_{ik} + \sum_{j\in\mathcal{N}}\mu_{jik} - \sum_{j\in\mathcal{N}}\mu_{ijk} \le 0,~\forall i,k
\]
\[
\mu_{ijk} - \widetilde{f}_{ijk} = 0,~\forall i,j,k,
\]
where $\mu_{ijk}=\mathbb{E}[\mu_{ijk}(t)]$ is the expected flow transmission rate under the $\omega$-only policy $\widehat{\pi}_c$ or $\widehat{\pi}_u$.
\end{lemma}
\begin{proof}
By our assumption, there exists a sequence of queue-respecting routing vectors $\{\mathbf{r}(t)\}_{t\ge 0}$ such that
\[
r_{ijk}(t) = \widetilde{f}_{ijk}(Q^*_{ik}(t)),~\forall t,~\forall i,j,k.
\]
By Lemma 2 in \cite{neely-stability}, there exists an $\omega$-only controllable policy $\widehat{\pi}_c: \omega_t\mapsto\bm{\mu}^c(t)$ and an $\omega$-only uncontrollable policy $\widehat{\pi}_u: \omega_t\mapsto \bm{\mu}^u(t)$ such that
\[
\mathbb{E}[\mu_{ijk}(t)] = \lim_{T\rightarrow\infty}\frac{1}{T}\sum_{t=0}^{T-1} \mathbb{E}[r_{ijk}(t)],~\forall i,j,k.
\]
As a result, we have
\[
\mu_{ijk}=\mathbb{E}[\mu_{ijk}(t)] =\lim_{T\rightarrow\infty}\frac{1}{T}\sum_{t=0}^{T-1} \mathbb{E}[r_{ijk}(t)] = \widetilde{f}_{ijk},~\forall i,j,k.
\]
By \eqref{eq:loaded} we have
\[
\lambda_{ik} +  \sum_{j\in\mathcal{N}} \mu_{jik}-\sum_{j\in\mathcal{N}}\mu_{ijk} \le 0,~\forall i,k,
\]
which completes the proof.
\end{proof}

\vspace{2mm}

Next we prove that the two virtual queues can be stabilized by the TMW algorithm.
\begin{lemma}\label{lm:TMW-stochastic-stability}
Under the TMW algorithm, we have
\[
\lim_{t\rightarrow\infty}\frac{\sum_{i,k} \mathbb{E}[X_{ik}(t)]}{t} =0,
\]
\[
\lim_{t\rightarrow\infty}\frac{\sum_{i,j,k} \mathbb{E}[|Y_{ijk}(t)|]}{t} =0.
\]
\end{lemma}
\begin{proof}
Define the Lyapunov function as 
\[
\Phi(t) = \sum_{i,k} X_{ik}^2(t) + \sum_{i,j,k}Y_{ijk}^2(t).
\]
By the evolution of virtual queue $\mathbf{X}(t)$, we have
\[
\sum_{i,k}X_{ik}^2(t+1)-\sum_{i,k}X_{ik}^2(t)\le \sum_{i,k}2X_{ik}(t)\Big(a_{ik}(t) + \sum_{j\in\mathcal{N}} g_{jik}(t) - \sum_{j\in\mathcal{N}}g_{ijk}(t)\Big)+N^3KD^2,
\]
where the inequality is due to the boundedness assumptions that $a_{ik}(t)\le D$ and $g_{ijk}(t)\le D$ for any $i,j,k$.  Similarly, for virtual queue $\mathbf{Y}(t)$ we have
\[
\sum_{i,j,k}Y^2_{ijk}(t+1)-\sum_{i,j,k}Y^2_{ijk}(t)\le \sum_{i,j,k}2Y_{ijk}(t)\Delta_{ijk}\Big(g_{ijk}(t), f_{ijk}(t), Q_{ik}(t)\Big)+N^2KD^2,
\]
where we rewrite $\Delta_{ijk}(t) = g_{ijk}(t) - \widetilde{f}_{ijk}(Q_{ik}(t))$ as $\Delta_{ijk}\big(g_{ijk}(t), f_{ijk}(t), Q_{ik}(t)\big)$ to explicitly emphasize on the dependence of $\Delta_{ijk}(t)$ on the imagined routing action $g_{ijk}(t)$, the true routing action $f_{ijk}(t)$ and the current queue length $Q_{ik}(t)$.

As a result, the conditional expected Lyapunov drift $\delta(t)\triangleq \mathbb{E}[\Phi(t+1)-\Phi(t)|\mathbf{X}(t),\mathbf{Y}(t)]$ can be bounded by
\[
\begin{split}
\delta(t) &\le 2\mathbb{E}\Big[\sum_{i,k}X_{ik}(t)\Big(a_{ik}(t) + \sum_{j\in\mathcal{N}} g_{jik}(t) - \sum_{j\in\mathcal{N}}g_{ijk}(t)\Big)+\sum_{i,j,k}Y_{ijk}(t)\Delta_{ijk}\Big(g_{ijk}(t), f_{ijk}(t), Q_{ik}(t)\Big)\Big]+2N^3KD^2\\
& \le 2\mathbb{E}\Big[\sum_{i,k}X_{ik}(t)\Big(a_{ik}(t) + \sum_{j\in\mathcal{N}} \mu_{jik}(t) - \sum_{j\in\mathcal{N}}\mu_{ijk}(t)\Big)+\sum_{i,j,k}Y_{ijk}(t)\Delta_{ijk}\Big(\mu_{ijk}(t), f_{ijk}(t), Q_{ik}(t)\Big)\Big]+2N^3KD^2,
\end{split}
\]
where the second inequality is due to the operation of TMW and $\bm{\mu}(t)$ is given in Lemma  \ref{lm:loaded-stochastic} which shows that
\[
\mathbb{E}\Big[a_{ik}(t) + \sum_{j\in\mathcal{N}} \mu_{jik}(t) - \sum_{j\in\mathcal{N}}\mu_{ijk}(t)\Big] = \lambda_{ik} +\sum_{j\in\mathcal{N}} \mu_{jik}-\sum_{j\in\mathcal{N}} \mu_{ijk}=0,~\forall i,k
\]
In addition, for any $i,j,k$, when there are enough queue backlogs ($Q_{ik}(t)\ge ND$) we have
\[
\begin{split}
\mathbb{E}\Big[\Delta_{ijk}\Big(\mu_{ijk}(t), f_{ijk}(t), Q_{ik}(t)\Big)\Big]& = \mathbb{E}\Big[\mu_{ijk}(t)-\widetilde{f}_{ijk}(Q_{ik}(t))\Big]\\
&=\mathbb{E}\Big[\mu_{ijk}(t)-f_{ijk}(t)\Big]\\
&\le  \mathbb{E}\Big[\mu_{ijk}(t)-\widetilde{f}_{ijk}(Q^*_{ik}(t))\Big]\\
&=0,
\end{split}
\]
where $\mathbf{Q}^*(t)$ is the queue length vector under the optimal controllable policy and the last equality is due to Lemma  \ref{lm:loaded-stochastic}. As a result, when there are enough backlogs (i.e., $Q_{ik}(t)\ge ND$ for any $i,k$) we have
\[
\delta(t)\le 2N^3KD^2.
\]
Let $t_0$ be the last time when there exists a queue $(i,k)$ such that $Q_{ik}(t)<ND$, i.e., $Q_{ik}\ge ND$ for any $t\ge t_0$ and any $i,k$. Without loss of generality, we assume $t_0=0$.  Summing over $t=0,\cdots,T-1$ and using law of iterated expectation, we have
\[
\mathbb{E}[\Phi(T)-\Phi(0)]\le 2N^3KD^2T.
\]
It follows that
\[
\mathbb{E}\Big[\sum_{i,k}X_{ik}^2(T)+\sum_{i,j,k}Y_{ijk}^2(T)\Big]\le \mathbb{E}\Big[\sum_{i,k}X_{ik}^2(0) + \sum_{i,j,k}Y_{ijk}^2(0)\Big]+2N^3KD^2T,
\]
which implies that
\[
\sum_{i,k}\mathbb{E}[X_{ik}(T)]\le \sqrt{NK}\sqrt{2\sum_{i,k}\mathbb{E}[X_{ik}^2(T)}]\le \sqrt{NK}\sqrt{\mathbb{E}\Big[\sum_{i,k}X_{ik}^2(0) + \sum_{i,j,k}Y_{ijk}^2(0)\Big]+2N^3KD^2T}.
\]
Assuming that $\sum_{i,k}X_{ik}(0)<\infty$ and $\sum_{i,j,k}Y_{ijk}(0)<\infty$, we have
\[
\lim_{T\rightarrow\infty}\frac{\sum_{i,k}\mathbb{E}[X_{ik}(T)]}{T}=0.
\]
Similarly, we have
\[
\lim_{T\rightarrow\infty}\frac{\sum_{i,j,k}\mathbb{E}[|Y_{ijk}(T)|]}{T}=0.
\]
This completes the proof.
\end{proof}

\vspace{2mm}

Finally, we show that as long as the two virtual queues $\mathbf{X}(t)$ and $\mathbf{Y}(t)$ are stable, then the physical queue $\mathcal{Q}(t)$ is also stable.
\begin{lemma}\label{lm:TMW-queue}
For any $i,k$ and $t\ge 0$, we have $Q_{ik}(t)\le X_{ik}(t)+\sum_{j\in\mathcal{N}}Y_{ijk}(t)-\sum_{j\in\mathcal{U}}Y_{jik}(t)$.
\end{lemma}
\begin{proof}
We prove this lemma by induction on $t$. The base case trivially holds true since we initialize $X_{ik}(0)=Q_{ik}(0)$ and $Y_{ijk}(0)=0$. Now suppose that $Q_{ik}(t)\le X_{ik}(t)+\sum_{j\in\mathcal{N}}Y_{ijk}(t)-\sum_{j\in\mathcal{U}}Y_{jik}(t)$ holds true for some $t\ge 0$. Then for any $i\in\mathcal{U}$ we have
\[
\begin{split}
Q_{ik}(t+1) & = Q_{ik}(t) + a_{ik}(t) + \sum_{j\in\mathcal{C}} \widetilde{g}_{jik}(Q_{jk}(t)) + \sum_{j\in\mathcal{U}} \widetilde{f}_{jik}(Q_{jk}(t)) - \sum_{j\in\mathcal{N}}\widetilde{f}_{ijk}(Q_{ik}(t))\\
&\le X_{ik}(t)+ a_{ik}(t)+\sum_{j\in\mathcal{N}} g_{jik}(t) -  \sum_{j\in\mathcal{N}} g_{ijk}(t) \\
&~~+\sum_{j\in\mathcal{C}} \widetilde{g}_{jik}(Q_{jk}(t))-\sum_{j\in\mathcal{C}} g_{jik}(t) \\
&~~+ \sum_{j\in\mathcal{N}}Y_{ijk}(t)+\sum_{j\in\mathcal{N}} g_{ijk}(t)-\sum_{j\in\mathcal{N}}\widetilde{f}_{ijk}(Q_{ik}(t))\\
&~~-\sum_{j\in\mathcal{U}}Y_{jik}(t) + \sum_{j\in\mathcal{U}} \widetilde{f}_{jik}(Q_{jk}(t))-\sum_{j\in\mathcal{U}} g_{jik}(t)\\
&\le X_{ik}(t+1)+\sum_{j\in\mathcal{N}}Y_{ijk}(t)-\sum_{j\in\mathcal{U}}Y_{jik}(t)+\sum_{j\in\mathcal{N}}\Delta_{ijk}(t)-\sum_{j\in\mathcal{U}}\Delta_{jik}(t)\\
&\le X_{ik}(t+1) + \sum_{j\in\mathcal{N}} Y_{ijk}(t+1)-\sum_{j\in\mathcal{U}}Y_{jik}(t+1),
\end{split}
\]
where the first inequality is due to the induction and simple algebra.
Similar induction applies to every controllable node $i\in\mathcal{C}$. This completes the induction proof.
\end{proof}

\vspace{2mm}

By Lemma \ref{lm:TMW-stochastic-stability} and Lemma \ref{lm:TMW-queue}, we can finally conclude that
\[
\begin{split}
\lim_{t\rightarrow\infty}\frac{\sum_{i,k} \mathbb{E}[Q_{ik}(t)]}{t}&\le \lim_{t\rightarrow\infty}\frac{\sum_{i,k} \Big(\mathbb{E}[X_{ik}(t)]+\sum_{j\in\mathcal{N}}\mathbb{E}[Y_{ijk}(t)]-\sum_{j\in\mathcal{U}}\mathbb{E}[Y_{jik}(t)]\Big)}{t}\\
&\le \lim_{t\rightarrow\infty}\frac{\sum_{i,k} \Big(\mathbb{E}[X_{ik}(t)] + \sum_{j\in\mathcal{U}}|\mathbb{E}[|Y_{ijk}(t)|] +  \sum_{j\in\mathcal{U}}\mathbb{E}[|Y_{jik}(t)|]\Big)}{t}\\
& = 0,
\end{split}
\]
which completes the proof to Theorem \ref{thm:TMW}.

\section{Extended Value Iteration}\label{ap:EVI}
In this appendix, we introduce Extended Value Iteration (EVI) \cite{UCRL2} as one of the approaches for optimistic planning (see step 3 in Algorithm \ref{alg:RL}). Extended Value Iteration is similar to the canonical Value Iteration but applies to MDPs with ``extended action space" where the additional action is the selection of the optimistic MDP (more precisely, the selection of the optimistic transition matrix) that yields the minimum average cost among a set of plausible MDPs. In particular, let $w_j(\mathbf{Q})$ be the value function for state $\mathbf{Q}$ obtained after iteration $j$. Then Extended Value Iteration proceeds as follows. For any $\mathbf{Q}\in\mathcal{Q}_V$
\begin{equation}\label{eq:EVI}
\begin{split}
w_0(\mathbf{Q})&=0\\
w_{j+1}(\mathbf{Q}) &= \min_{\alpha\in\mathcal{A}}\Big\{\sum_{i,k} Q_{ik}  + \min_{p(\cdot)\in \mathcal{P}_{\ell}(\mathbf{Q},\alpha)}\Big\{\sum_{\mathbf{Q}'\in\mathcal{Q}_V}p(\mathbf{Q}')w_j(\mathbf{Q}')\Big\}\Big\},
\end{split}
\end{equation}
where $\mathcal{P}_{\ell}(\mathbf{Q},\alpha)$ is given in \eqref{eq:confidence-set}.
The inner minimization problem can be easily solved  using the following greedy algorithm which tries to push as much probability mass to the state with the largest value.
\begin{framed}
\noindent \textbf{Greedy Algorithm:} find optimistic transition probability $p(\cdot)$ for a particular state-action pair $(\mathbf{Q},\alpha)$

\vspace{2mm}

\noindent \textbf{Input:} Estimated transition probability $\hat{P}(\cdot|\mathbf{Q},\alpha)$ and distance $d(\mathbf{Q},\alpha)$ as given in \eqref{eq:confidence-set} 

\vspace{2mm}

1. Sort states in $\mathcal{Q}_V$ in descending order according to their $w_j(\cdot)$ value. We assume that the sorted state space is denoted by $\mathcal{Q}'_V=\{\mathbf{Q}_1',\mathbf{Q}_2',\cdots,\mathbf{Q}_n'\}$ where
\[
w_j(\mathbf{Q}_1')\ge w_j(\mathbf{Q}_2')\ge \cdots \ge w_j(\mathbf{Q}_n').
\]

2. Set
\[
p(\mathbf{Q}'_1) = \min\Big\{1,\hat{P}(\mathbf{Q}_1'|\mathbf{Q},\alpha)+\frac{d(\mathbf{Q},\alpha)}{2}\Big\}
\]
\[
p(\mathbf{Q}'_i) = \hat{P}(\mathbf{Q}_i'|\mathbf{Q},\alpha)\text{ for all states }\mathbf{Q}_i'~\text{with }i>1.
\]

3. Set $m=n$

4. \textbf{While} $\sum_{\mathbf{Q}'\in\mathcal{Q}'_V}p(\mathbf{Q}')>1$ \textbf{do}

\qquad (a) Reset $p(\mathbf{Q}'_m)=\max\{0,1-\sum_{\mathbf{Q}'\ne \mathbf{Q}'_m}p(\mathbf{Q}')\}$.

\qquad (b) Set $m=m-1$.
\end{framed}
The distance $d(\mathbf{Q},\alpha)$  in the input is the range used in our confidence set construction  \eqref{eq:confidence-set}, i.e., 
\[
d(\mathbf{Q},\alpha)=\sqrt{\frac{C \log(2|\mathcal{A}|t_{\ell} V)}{\max\{1,n_{\ell}(\mathbf{Q},\alpha)\}}}.
\]

The iteration stops when the following convergence criterion is triggered:
\begin{equation}\label{eq:convergence}
\max_{\mathbf{Q}\in\mathcal{Q}_V} \{w_{j+1}(\mathbf{Q})-w_j(\mathbf{Q})\} - \min_{\mathbf{Q}\in\mathcal{Q}_V} \{w_{j+1}(\mathbf{Q})-w_j(\mathbf{Q})\} \le \frac{1}{\sqrt{t_{\ell}}},
\end{equation}
which requires that the span of the average cost across all states be less than or equal to $\frac{1}{\sqrt{t_{\ell}}}$.

\section{Proof to Theorem \ref{thm:tucrl-final}}\label{ap:tucrl-final}
We first introduce some notations.
\begin{itemize}
\item{Let $M(P)$ be an MDP model for the original untruncated system (where no packet dropping happens) with transition matrix $P$. 
Denote by $M^*\triangleq M(P^*)$ the true MDP model for the original untruncated system with the true transition matrix $P^*$. Note that in an untruncated system, the state space is the countably-infinite  queue length space $\mathcal{Q}$.}
\item{Let $M_V(P)$ be a \emph{truncated MDP} model with transition matrix $P$ for a truncated system with threshold $V$. Note that the state space in a truncated MDP model is the truncated queue length  space $\mathcal{Q}_V$. Denote by $M^*_V$ the \emph{truncated true MDP model} that is obtained by applying packet dropping to the original untruncated true MDP model $M^*$. 
}
\item{Let $J^*$ be the optimal time-average expected queue length in the untruncated true  MDP $M^*$. Similarly, denote by $J_V^*$ the optimal time-average expected queue length in the truncated true MDP $M_V^*$.}
\end{itemize}

The following theorem provides an upper bound on the total queue length in a truncated system with threshold $V$.
\begin{theorem}\label{thm:truncated-bound}
With probability at least $1-O\Big(\frac{1}{T}+\frac{1}{V}\Big)$, the TUCRL algorithm achieves
\begin{equation}\label{eq:final-truncated}
\begin{split}
\sum_{t=0}^{T-1}\sum_{i,k} Q_{ik}(t) &= TJ_V^* + \gamma(T,V),
\end{split}
\end{equation}
where
\[
\gamma(T,V)\triangleq O\Big(V^{N+2}\log T + V^{2+N\slash 2}\sqrt{ T\log(TV)}\Big).
\]
\end{theorem}
\begin{proof}
See Appendix \ref{ap:truncated-bound}.
\end{proof}

\vspace{1mm}

Next, we provide an upper bound on the fraction of dropped packets based on Theorem \ref{thm:truncated-bound}.

\begin{theorem}\label{thm:drop}
When $T$ is sufficiently large, with probability at least $1-O\Big(\frac{1}{T}+\frac{1}{V}\Big)$, the fraction of dropped packets is
\[
\eta_T = O\Big(\frac{\gamma(T,V)}{TV}+\frac{J_V^*}{V}\Big).
\]
\end{theorem}
\begin{proof}
See Appendix \ref{ap:drop}.
\end{proof}

\vspace{1mm}

Note that the above results depend on the optimal average queue length $J_V^*$ in the truncated system. The following lemma shows the relationship between $J_V^*$ and $J^*$ (the optimal average queue length in the original untruncated system).
\begin{lemma}\label{lm:cc}
$\mathbb{P}[J_V^*=J^*] \ge 1-O\Big(\frac{1}{V}\Big)$.
\end{lemma}
\begin{proof}
See Appendix \ref{ap:cc}.
\end{proof}

\vspace{1mm} 

Combining Lemma \ref{lm:cc}, Theorem \ref{thm:truncated-bound} and Theorem \ref{thm:drop}, we have that with probability at least $1-\Theta\Big(\frac{1}{T}+\frac{1}{V}\Big)$ the TUCRL algorithm achieves
\[
\frac{1}{T}\sum_{t=0}^{T-1}\sum_{i,k} Q_{ik}(t) \le J^* + \frac{\gamma(T,V)}{T}
\]
and 
\[
\eta_T \le O\Big(\frac{\gamma(T,V)}{TV}+\frac{J^*}{V}\Big).
\]
On the other hand, the worst-case average queue length is at most $V$ and the worst-case packet dropping rate is 1, which happens with probability at most $\Theta\Big(\frac{1}{T}+\frac{1}{V}\Big)$. As a result, the expected average queue length under TUCRL is
\[
\begin{split}
\frac{1}{T}\sum_{t=0}^{T-1}\sum_{i,k} \mathbb{E}[Q_{ik}(t)]&\le \Big[1-\Theta\Big(\frac{1}{T}+\frac{1}{V}\Big)\Big]\Big(J^* + \frac{\gamma(T,V)}{T}\Big) + \Theta\Big(\frac{1}{T}+\frac{1}{V}\Big)V
\end{split}
\]
and the expected fraction of dropped packets is 
\[
\begin{split}
\mathbb{E}[\eta_T] &\le \Big[1-\Theta\Big(\frac{1}{T}+\frac{1}{V}\Big)\Big]O\Big(\frac{\gamma(T,V)}{TV}+\frac{J^*}{V}\Big) + \Theta\Big(\frac{1}{T}+\frac{1}{V}\Big)
\end{split}
\]
For any given $V<\infty$, taking $T\rightarrow\infty$ and noting that $J^*<\infty$ (see Assumption \ref{as:mdp-stability}) yield the desirable performance bounds in Theorem \ref{thm:tucrl-final}.

\section{Number of Episodes in TUCRL}\label{as:epi}
The following lemma provides an upper bound on the number of episodes in TUCRL.
\begin{lemma}\label{lm:episode}
The total number of episodes in TUCRL up to time $T\ge 1$ is bounded by
\[
L_T\le 1+|\mathcal{Q}_V||\mathcal{A}|(\log_2 T+1).
\]
\end{lemma}
\begin{proof}
Let $\mathcal{L}(\mathbf{Q},\alpha)$ be the set of episodes where the number of visits to a state-action pair $(\mathbf{Q},\alpha)$ is doubled up to the beginning of time $T$, i.e.,
\[
\mathcal{L}(\mathbf{Q},\alpha) = \Big\{\ell< L_T:v_{\ell}(\mathbf{Q},\alpha)\ge \max\{1,n_{\ell}(\mathbf{Q},\alpha)\}\Big\},
\]
where $L_T$ is the index of the episode that contains time slot $T$. Also let $\mathcal{K}=\{(\mathbf{Q},\alpha):n_{L_T}(\mathbf{Q},\alpha)\ge 1\}$ be the set of state-action pairs that have non-zero visits up to the beginning of episode $L_T$. Then we claim that $|\mathcal{L}(\mathbf{Q},\alpha)|\le \log_2 \Big(n_{L_T}(\mathbf{Q},\alpha)\Big)+1$ for any $(\mathbf{Q},\alpha)\in \mathcal{K}$. Indeed, if $|\mathcal{L}(\mathbf{Q},\alpha)|> \log_2 \Big(n_{L_T}(\mathbf{Q},\alpha)\Big)+1$, for any $(\mathbf{Q},\alpha)\in\mathcal{K}$ we have
\[
n_{L_T}(\mathbf{Q},\alpha)=\prod_{\ell=1}^{L_T-1}\frac{\max\{1,n_{\ell+1}(\mathbf{Q},\alpha)\}}{\max\{1,n_{\ell}(\mathbf{Q},\alpha)\}}\ge \prod_{\ell \in \mathcal{L}(\mathbf{Q},\alpha)}\frac{n_{\ell+1}(\mathbf{Q},\alpha)}{\max\{1,n_{\ell}(\mathbf{Q},\alpha)\}}
\ge 2^{|\mathcal{L}(\mathbf{Q},\alpha)|-1}>n_{L_T}(\mathbf{Q},\alpha),
\]
which leads to a contradiction. As a result, we have
\[
\begin{split}
L_T \le \sum_{(\mathbf{Q},\alpha)\in \mathcal{K}}|\mathcal{L}(\mathbf{Q},\alpha)| + 1&\le \sum_{(\mathbf{Q},\alpha)\in \mathcal{K}}\log_2 \Big(n_{L_T}(\mathbf{Q},\alpha)\Big)+|\mathcal{K}|+1\\
& \le |\mathcal{K}|\log_2 \Big(\frac{\sum_{(\mathbf{Q},\alpha)\in\mathcal{K}} n_{L_T}(\mathbf{Q},\alpha)}{ |\mathcal{K}|}\Big)+|\mathcal{K}|+1\\
&\le |\mathcal{K}|\log_2\Big(\frac{T}{|\mathcal{K}|}\Big)+|\mathcal{K}|+1\\
&\le |\mathcal{K}|\log_2 T+|\mathcal{K}|+1\\
&\le |\mathcal{Q}_V||\mathcal{A}|(\log_2 T+1)+1,
\end{split}
\]
where the second inequality is due to the concavity of $\log$ function.
\end{proof}

\section{Proof to Theorem \ref{thm:truncated-bound}}\label{ap:truncated-bound}
Define the regret (w.r.t. queue length) in episode $\ell$ as 
\[
\Delta_{\ell} \triangleq \sum_{t=t_{\ell}}^{t_{\ell+1}-1}\Big(\sum_{i,k} Q_{ik}(t)-J_V^*\Big),
\]
where  $J_V^*$ is the optimal average queue length for the truncated true MDP $M_V^*$. Let $L_T$ be the number of episodes up to time $T$ (including the one that contains $T$). Without loss of generality, we assume that $t_{L_T+1}=T$ (i.e., $T-1$ is the end of episode $L_T$). Also denote by $\mathcal{L}_1=\{\ell\le L_T: t_{\ell+1}<\log_2 T\}$ the set of episodes whose ending time is less than $\log_2 T$, and we similarly define $\mathcal{L}_2=\{\ell\le L_T:t_{\ell+1}\ge \log_2 T\}$. The regret up to time $T$ can be decomposed by
\begin{equation}\label{eq:dp}
\begin{split}
\sum_{\ell=1}^{L_T} \Delta_{\ell} &= \sum_{\ell\in\mathcal{L}_1} \Delta_{\ell} + \sum_{\ell	\in\mathcal{L}_2} \Delta_{\ell} \le \log_2 T(NKD\log_2 T) + \sum_{\ell\in\mathcal{L}_2} \Delta_{\ell},
\end{split}
\end{equation}
where the inequality is due to the fact the total amount of exogenous arrivals (and thus the total queue length) within the first $\log_2 T$ slots is at most $NKD\log_2T$ by our boundedness assumption.
It turns our that the first term in \eqref{eq:dp} is of a lower order, and we focus on finding an upper bound on the second term. The following lemma shows that the constructed confidence set $\mathcal{M}_{\ell}$ contains the truncated true MDP model $M_V^*$ with high probability. 
\begin{lemma}\label{lm:confidence}
For any $T\ge 2$, the probability that the truncated true MDP $M_V^*$ belongs to the confidence set $\mathcal{M}_{\ell}$ for every episode $\ell\in\mathcal{L}_2$ is
\[
\mathbb{P}[M_V^* \in \mathcal{M}_{\ell},~\forall \ell\in\mathcal{L}_2] \ge 1-\Big(\frac{1}{V^{N+1}|\mathcal{A}|}+\frac{2}{V}\Big).
\]
\end{lemma}
\begin{proof}
See Appendix \ref{ap:confidence}.
\end{proof}

\vspace{2mm}

In the following, we only consider the case where $M_V^* \in \mathcal{M}_{\ell},~\forall \ell\in\mathcal{L}_2$ since this holds with high probability by Lemma \ref{lm:confidence}.
Let $M_{\ell}\triangleq M_V(P_{\ell})$ be the optimistic MDP selected in episode $\ell$ and define $J_{\ell}\triangleq J^{\pi_{\ell}}(M_{\ell})$ as the average queue length under MDP $M_{\ell}$ and the nearly-optimal policy $\pi_{\ell}$. Note that $J_{\ell} - J^*(M_{\ell}) \le \frac{1}{\sqrt{t_{\ell}}}$ by the accuracy requirement in Step \ref{step:plan} of the TUCRL algorithm. Moreover, due to the fact that the truncated true MDP $M_V^*\in\mathcal{M}_{\ell}$ and $M_{\ell}$ is the optimistic MDP that yields the minimum average queue length among all MDPs in $\mathcal{M}_{\ell}$, we have $J^*(M_{\ell}) \le J_V^*$ and thus
\begin{equation}\label{eq:opt-gap}
J_{\ell} - J_V^* \le J_{\ell} - J^*(M_{\ell}) \le \frac{1}{\sqrt{t_{\ell}}}.
\end{equation}
Then we can decompose $\Delta_{\ell}$ as follows:
\[
\begin{split}
\Delta_{\ell} &= \sum_{t=t_{\ell}}^{t_{\ell+1}-1}\Big(\sum_{i,k} Q_{ik}(t)-J_{\ell} + J_{\ell}-J_V^*)\Big)\\
&\le \sum_{t=t_{\ell}}^{t_{\ell+1}-1}\Big(\sum_{i,k} Q_{ik}(t)-J_{\ell}\Big) + \frac{t_{\ell+1}-t_{\ell}}{\sqrt{t_{\ell}}}.
\end{split}
\]
Summing over all episodes $\ell\in\mathcal{L}_2$, we have
\begin{equation}\label{eq:dd}
\begin{split}
\sum_{\ell\in\mathcal{L}_2} \Delta_{\ell} &\le   \sum_{\ell\in\mathcal{L}_2} \sum_{t=t_{\ell}}^{t_{\ell+1}-1}\Big(\sum_{i,k} Q_{ik}(t)-J_{\ell}\Big)+ \sum_{\ell\in\mathcal{L}_2} \frac{t_{\ell+1}-t_{\ell}}{\sqrt{t_{\ell}}}.
\end{split}
\end{equation}

We first find an upper bound for the second term in the RHS of \eqref{eq:dd} by using Lemma 19 in \cite{UCRL2}: for any sequence of numbers $z_1,\cdots,z_n$ with $0\le z_\ell\le Z_{\ell-1}:=\max\Big\{1,\sum_{i=1}^{\ell-1}z_i\Big\}$, we have
\[
\sum_{\ell=1}^n\frac{z_\ell}{\sqrt{Z_{\ell-1}}}\le (\sqrt{2}+1)\sqrt{Z_n}.
\]
As a result, defining $z_\ell = t_{\ell+1}-t_{\ell}$ and $Z_{\ell-1}=\max\Big\{1,t_{\ell}\Big\}$, we have
\begin{equation}\label{eq:sqrt}
\sum_{\ell\in\mathcal{L}_2} \frac{t_{\ell+1}-t_{\ell}}{\sqrt{t_{\ell}}}\le \sum_{\ell=1}^{L_T}\frac{t_{\ell+1}-t_{\ell}}{\sqrt{\max\Big\{1,t_{\ell}\Big\}}}\le (\sqrt{2}+1)\sqrt{t_{L_T+1}}\le (\sqrt{2}+1)\sqrt{T}=O(\sqrt{T}).
\end{equation}
Then it remains to find an upper bound on the first term in the RHS of \eqref{eq:dd}. 
\begin{lemma}\label{lm:dd1}
If $M_V^*\in\mathcal{M}_{\ell}$ for any $\ell\in\mathcal{L}_2$, then with probability at least $1-\frac{1}{T}$ we have
\[
\sum_{\ell\in\mathcal{L}_2} \sum_{t=t_{\ell}}^{t_{\ell+1}-1}\Big(\sum_{i,k} Q_{ik}(t)-J_{\ell}\Big) = O\Big(V^{N+2}\log T + V^{2+N\slash 2}\sqrt{ T\log(TV)}\Big).
\]
\end{lemma}
\begin{proof}
See Appendix \ref{ap:dd1}.
\end{proof}

\vspace{2mm}

Combining \eqref{eq:dp}, \eqref{eq:sqrt}, \eqref{eq:dd}, Lemma \ref{lm:confidence} and Lemma \ref{lm:dd1}, we can conclude that with probability at least $1-\Big(\frac{1}{T}+\frac{1}{V^{N+1}|\mathcal{A}|}+\frac{2}{V}\Big)$
\[
\sum_{\ell=1}^{L_T} \Delta_{\ell}  = O\Big(V^{N+2}\log T + V^{2+N\slash 2}\sqrt{ T\log(TV)}\Big),
\]
which completes the proof to Theorem \ref{thm:truncated-bound}.

\section{Proof to Lemma \ref{lm:confidence}}\label{ap:confidence}
The $\ell_1$-deviation of the true distribution $p(\cdot)$ and the empirical distribution $\hat{p}(\cdot)$ over $m$ distinct events from $n$ samples is bounded according to Weissman \emph{et al.} \cite{l1-inequality} by
\[
\mathbb{P}\Big[||\hat{p}(\cdot)-p(\cdot)||_1 \ge \epsilon\Big] \le (2^m-2)\exp\Big(-\frac{n\epsilon^2}{2}\Big).
\]
In our case, for each state-action pair $(\mathbf{Q},\alpha)$, the number of possible next states is at most $m=(2ND+1)^N$ since the increase/decrease in the length of each queue is at most $ND$ in every slot (including both exogenous and endogenous arrivals). Let $P_V^*$ be the transition probabilities for the truncated true MDP $M^*_V$. Then the number of samples we obtained for distribution $P_V^*(\cdot|\mathbf{Q},\alpha)$ up to the beginning of episode $\ell$ (i.e., when $t=t_{\ell}$) is $n_{\ell}(\mathbf{Q},\alpha)$. As a result, for episode $\ell$ and any state-action pair $(\mathbf{Q},\alpha)$ such that $n_{\ell}(\mathbf{Q},\alpha)\ge 1$, setting 
\[
\epsilon = \sqrt{\frac{2\log(2^m V^{2N+1} |\mathcal{A}|^2 t_{\ell})}{n_{\ell}(\mathbf{Q},\alpha)}} = \sqrt{\frac{2\log(2^m V^{2N+1} |\mathcal{A}|^2 t_{\ell})}{\max\{1,n_{\ell}(\mathbf{Q},\alpha)\}}} \le \sqrt{\frac{C\log(2|\mathcal{A}| t_{\ell} V)}{\max\{1,n_{\ell}(\mathbf{Q},\alpha)\}}}\quad (C\triangleq 2m),
\]
guarantees that 
\begin{equation}\label{eq:bbb}
\mathbb{P}\Big[\Big|\Big|\hat{P}(\cdot|\mathbf{Q},\alpha)-P_V^*(\cdot|\mathbf{Q},\alpha)\Big|\Big|_1 \ge \sqrt{\frac{C\log(2|\mathcal{A}| t_{\ell} V)}{\max\{1,n_{\ell}(\mathbf{Q},\alpha)\}}}\Big]\le \frac{1}{t_{\ell} V^{2N+1}|\mathcal{A}|^2 }.
\end{equation}
Note that the above bound holds for every state-action pair $(\mathbf{Q},\alpha)$ such that $n_{\ell}(\mathbf{Q},\alpha)\ge 1$.
If there  have not been any observations for some state-action pair $(\mathbf{Q},\alpha)$, i.e., $n_{\ell}(\mathbf{Q},\alpha)=0$, the confidence intervals \eqref{eq:confidence-set} trivially hold with probability 1. Thus, inequality \eqref{eq:bbb} actually holds for every  $\mathbf{Q}\in\mathcal{Q}_V$ and $\alpha \in \mathcal{A}$.
 
Applying the union bound over all possible state-action pairs in the truncated MDP, we have for each episode $\ell$
\[
\mathbb{P}[P_V^* \notin \mathcal{P}_{\ell}] \le \frac{1}{t_{\ell}V^{N+1}|\mathcal{A}|},
\]
i.e.,
\begin{equation}\label{eq:confidence-1}
\mathbb{P}[M_V^* \notin \mathcal{M}_{\ell}]\le \frac{1}{t_{\ell}V^{N+1}|\mathcal{A}|}.
\end{equation}
Let $\ell'$ be the smallest episode index in $\mathcal{L}_2$. For any $\ell\in\mathcal{L}_2\slash \{\ell'\}$, we have $t_{\ell} \ge \log_2 T$. As a result, we have for any $T\ge 2$
\[
\begin{split}
\mathbb{P}[M_V^* \notin \mathcal{M}_{\ell},~\exists \ell\in\mathcal{L}_2]&\le \sum_{\ell\in\mathcal{L}_2}\frac{1}{t_{\ell}V^{N+1}|\mathcal{A}|}\\
& \le \frac{1}{t_{\ell'}V^{N+1}|\mathcal{A}|} + \frac{|\mathcal{L}_2|-1}{V^{N+1}|\mathcal{A}|\log_2 T}\\
& \le \frac{1}{V^{N+1}|\mathcal{A}|} + \frac{V^N |\mathcal{A}|(\log_2 T+1)}{V^{N+1}|\mathcal{A}|\log_2 T}\\
& = \frac{1}{V^{N+1}|\mathcal{A}|} + \frac{2}{V},
\end{split}
\]
where the first inequality is derived by applying union bound to \eqref{eq:confidence-1}, the second inequality is due to the definition of $\mathcal{L}_2$, the third inequality is due to Lemma \ref{lm:episode} which states that the total number of episodes up to time $T$ is at most $1+|\mathcal{Q}_V||\mathcal{A}|(\log_2 T+1)$.

\section{Proof to Lemma \ref{lm:dd1}}\label{ap:dd1}
Note that $\sum_{t=t_{\ell}}^{t_{\ell+1}-1}\sum_{i,k} Q_{ik}(t)$ can be re-arranged according to the visiting frequency $v_{\ell}(\mathbf{Q},\alpha)$ for each state-action pair $(\mathbf{Q},\alpha)$ during episode $\ell$:
\[
\sum_{t=t_{\ell}}^{t_{\ell+1}-1} \sum_{i,k} Q_{ik}(t) = \sum_{(\mathbf{Q},\alpha)} v_{\ell}(\mathbf{Q},\alpha) \sum_{i,k} Q_{ik}.
\]
As a result, we have
\begin{equation}\label{eq:dd22}
\sum_{\ell\in\mathcal{L}_2} \sum_{t=t_{\ell}}^{t_{\ell+1}-1}\Big(\sum_{i,k} Q_{ik}(t)-J_{\ell}\Big) =  \sum_{\ell\in\mathcal{L}_2} \sum_{(\mathbf{Q},\alpha)} v_{\ell}(\mathbf{Q},\alpha) \Big(\sum_{i,k} Q_{ik} - J_{\ell}\Big).
\end{equation}

We first provide an observation regarding the value functions we derived in Extended Value Iteration (EVI) \eqref{eq:EVI}.
\begin{observation}\label{ob:range}
If the truncated true MDP $M_V^*$ belongs to the confidence set $\mathcal{M}_{\ell}$, then for any iteration $j$ in EVI, the range of the derived value function $w_j(\cdot)$ is bounded by $O(V^2)$, i.e., 
\begin{equation}\label{eq:value-bound}
\max_{\mathbf{Q}\in\mathcal{Q}_V} w_j(\mathbf{Q}) - \min_{\mathbf{Q}\in\mathcal{Q}_V} w_j(\mathbf{Q}) \le LN^2V^2.
\end{equation}
\end{observation}
\begin{proof}
Note that in the $j$th iteration of EVI, the derived value function $w_j(\mathbf{Q})$ is the sum of expected queue lengths within $j$ time steps under an optimal (possibly non-stationary) $j$-step policy starting in state $\mathbf{Q}$ in the MDP with extended action set (the extra action is to select an MDP within $\mathcal{M}_{\ell}$). 
If there are two states $\mathbf{Q},\mathbf{Q}'\in\mathcal{Q}_V$ such that $w_j(\mathbf{Q})-w_j(\mathbf{Q}')>LN^2V^2$, then an improved value for state $\mathbf{Q}$ could be achieved by adopting the following policy. First, follow a policy which always selects the truncated true MDP $M_V^*$ (this is a feasible action in the extended action set by our assumption) and moves from $\mathbf{Q}$ to $\mathbf{Q}'$ most quickly in $M_V^*$, which takes at most $L||\mathbf{Q}'-\mathbf{Q}||_1\le LNV$ steps on average by Assumption \ref{as:diameter}. After reaching state $\mathbf{Q}'$, we follow the optimal $j$-step policy for $\mathbf{Q}'$ for the remaining $j-LNV$ time steps. Note that the sum of accumulative queue lengths before reaching state $\mathbf{Q}'$ is at most $LN^2V^2$ since in each time step the total queue length is at most $NV$ in the truncated system. As a result, the constructed policy yields $w_j(\mathbf{Q})\le  w_j(\mathbf{Q}')+LN^2V^2$, contradicting our assumption that $w_j(\mathbf{Q})-w_j(\mathbf{Q}')>LN^2V^2$. 
\end{proof}

\vspace{2mm}

It is a direct consequence of Theorem 8.5.6. of \cite{MDP}, that when the convergence criterion \eqref{eq:convergence} holds at iteration $j$, then
\begin{equation}\label{eq:stopping}
|w_{j+1}(\mathbf{Q})-w_{j}(\mathbf{Q}) - J_{\ell}|\le \frac{1}{\sqrt{t_{\ell}}},~\forall \mathbf{Q}\in\mathcal{Q}_V.
\end{equation}
In the following, letter $j$ will be reserved to reference the iteration where the convergence criterion \eqref{eq:convergence} is triggered. Expanding $w_{j+1}(\mathbf{Q})$ according to \eqref{eq:EVI}, we have
\[
w_{j+1}(\mathbf{Q}) = \sum_{i,k} Q_{ik} + \sum_{\mathbf{Q}'\in \mathcal{Q}_V}P_{\ell}\Big(\mathbf{Q}'|\mathbf{Q},\pi_{\ell}(\mathbf{Q})\Big)w_j(\mathbf{Q}'),~\forall \mathbf{Q}\in\mathcal{Q}_V,
\]
and hence by \eqref{eq:stopping} we have
\[
\Big|\Big(\sum_{i,k} Q_{ik} - J_{\ell}\Big) + \sum_{\mathbf{Q}'\in \mathcal{Q}_V}P_{\ell}\Big(\mathbf{Q}'|\mathbf{Q},\pi_{\ell}(\mathbf{Q})\Big)w_j(\mathbf{Q}')-w_j(\mathbf{Q})\Big|\le \frac{1}{\sqrt{t_{\ell}}},~\forall \mathbf{Q}\in\mathcal{Q}_V,
\]
which implies that for any  $\mathbf{Q}\in\mathcal{Q}_V$
\[
\begin{split}
\sum_{i,k} Q_{ik} - J_{\ell} &\le \frac{1}{\sqrt{t_{\ell}}} + w_j(\mathbf{Q}) - \sum_{\mathbf{Q}'\in \mathcal{Q}_V}P_{\ell}\Big(\mathbf{Q}'|\mathbf{Q},\pi_{\ell}(\mathbf{Q})\Big)w_j(\mathbf{Q}')\\
& = \frac{1}{\sqrt{t_{\ell}}} + w_j(\mathbf{Q}) - \min_{\mathbf{Q}\in\mathcal{Q}_V} w_j (\mathbf{Q})- \sum_{\mathbf{Q}'\in \mathcal{Q}_V}P_{\ell}\Big(\mathbf{Q}'|\mathbf{Q},\pi_{\ell}(\mathbf{Q})\Big)\Big(w_j(\mathbf{Q}')-\min_{\mathbf{Q}\in\mathcal{Q}_V} w_j (\mathbf{Q})\Big)，
\end{split}
\]
where the second equality holds because $\sum_{\mathbf{Q}'\in \mathcal{Q}_V}P_{\ell}\Big(\mathbf{Q}'|\mathbf{Q},\pi_{\ell}(\mathbf{Q})\Big)=1$. 

Define $h_{\ell}(\mathbf{Q})\triangleq w_j(\mathbf{Q})-\min_{\mathbf{Q}\in\mathcal{Q}_V} w_j (\mathbf{Q})$ (we use the subscript $\ell$ to reference the episode). It is clear from Observation \ref{ob:range} that $||h_{\ell}(\cdot)||_{\infty}\le LN^2V^2$. Then we can rewrite the above inequality as
\begin{equation}\label{eq:expand}
\sum_{i,k} Q_{ik} - J_{\ell} \le \frac{1}{\sqrt{t_{\ell}}} + h_{\ell}(\mathbf{Q}) - \sum_{\mathbf{Q}'\in \mathcal{Q}_V}P_{\ell}\Big(\mathbf{Q}'|\mathbf{Q},\pi_{\ell}(\mathbf{Q})\Big)h_{\ell}(\mathbf{Q}').
\end{equation}
As a result, we have
\begin{equation}\label{eq:delta-2}
\sum_{(\mathbf{Q},\alpha)} v_{\ell}(\mathbf{Q},\alpha) \Big(\sum_{i,k} Q_{ik} - J_{\ell}\Big)\le  \sum_{(\mathbf{Q},\alpha)} v_{\ell}(\mathbf{Q},\alpha)  \Big(\frac{1}{\sqrt{t_{\ell}}} + h_{\ell}(\mathbf{Q}) - \sum_{\mathbf{Q}'\in \mathcal{Q}_V}P_{\ell}\Big(\mathbf{Q}'|\mathbf{Q},\pi_{\ell}(\mathbf{Q})\Big)h_{\ell}(\mathbf{Q}')\Big).
\end{equation}
To simplify notations, we define a row vector $\mathbf{v}_{\ell}\triangleq \Big(v_{\ell}\big(\mathbf{Q},\pi_{\ell}(\mathbf{Q})\big)\Big)_{\mathbf{Q}\in\mathcal{Q}_V}$, a column vector $\mathbf{h}_{\ell} \triangleq \big(h_{\ell}(\mathbf{Q})\big)_{\mathbf{Q}\in\mathcal{Q}_V}$ and a matrix $\mathbf{G}_{\ell}=\Big(P_{\ell}\big(\mathbf{Q}'|\mathbf{Q},\pi_{\ell}(\mathbf{Q})\big)\Big)_{\mathbf{Q},\mathbf{Q}'\in\mathcal{Q}_V}$. Then \eqref{eq:delta-2} can be rewritten as
\begin{equation}\label{eq:delta-3}
\begin{split}
\sum_{(\mathbf{Q},\alpha)} v_{\ell}(\mathbf{Q},\alpha) \Big(\sum_{i,k} Q_{ik} - J_{\ell}\Big)&\le\sum_{(\mathbf{Q},\alpha)} \frac{v_{\ell}(\mathbf{Q},\alpha)}{\sqrt{t_{\ell}}} + \mathbf{v}_{\ell}(\mathbf{I}-\mathbf{G}_{\ell})\mathbf{h}_{\ell}\\
&=\frac{t_{\ell+1}-t_{\ell}}{\sqrt{t_{\ell}}} + \mathbf{v}_{\ell}(\mathbf{I}-\mathbf{G}_{\ell})\mathbf{h}_{\ell}.
\end{split}
\end{equation}
The second term in the RHS of \eqref{eq:delta-3} can be further decomposed as follows:
\begin{equation}\label{eq:delta-4}
\begin{split}
 \mathbf{v}_{\ell}(\mathbf{I}-\mathbf{G}_{\ell})\mathbf{h}_{\ell} &=  \mathbf{v}_{\ell}(\mathbf{I}-\mathbf{G}^*_{\ell}+\mathbf{G}^*_{\ell}-\mathbf{G}_{\ell})\mathbf{h}_{\ell}\\
 & = \mathbf{v}_{\ell}(\mathbf{I}-\mathbf{G}^*_{\ell})\mathbf{h}_{\ell} +  \mathbf{v}_{\ell}(\mathbf{G}^*_{\ell}-\mathbf{G}_{\ell})\mathbf{h}_{\ell},
 \end{split}
\end{equation}
where $\mathbf{G}^*_{\ell}$ is the transition matrix if we apply policy $\pi_{\ell}$ in the truncated true MDP model $M_V^*$. Summing over all episodes $\ell\in\mathcal{L}_2$, we have
\begin{equation}\label{eq:delta-5}
\sum_{\ell\in\mathcal{L}_2}\sum_{(\mathbf{Q},\alpha)} v_{\ell}(\mathbf{Q},\alpha) \Big(\sum_{i,k} Q_{ik} - J_{\ell}\Big) \le \sum_{\ell\in\mathcal{L}_2} \frac{t_{\ell+1}-t_{\ell}}{\sqrt{t_{\ell}}} + \sum_{\ell\in\mathcal{L}_2} \mathbf{v}_{\ell}(\mathbf{I}-\mathbf{G}^*_{\ell})\mathbf{h}_{\ell} +  \sum_{\ell\in\mathcal{L}_2} \mathbf{v}_{\ell}(\mathbf{G}^*_{\ell}-\mathbf{G}_{\ell})\mathbf{h}_{\ell}
\end{equation}
Simarly, by Lemma 19 in \cite{UCRL2}, the first term in the RHS of \eqref{eq:delta-5} can be bounded by $O(\sqrt{T})$, which is a  lower-order term. Next we find upper bounds for the last two terms in the RHS of \eqref{eq:delta-5}, given in the following two lemmas.

\begin{lemma}\label{lm:martingale}
$\sum_{\ell\in\mathcal{L}_2}  \mathbf{v}_{\ell}(\mathbf{I}-\mathbf{G}^*_{\ell})\mathbf{h}_{\ell} =  O(V^2\sqrt{T\log T}+V^{N+2}\log T)$ with probability at least $1-\frac{1}{T}$.
\end{lemma}
\begin{proof}
See Appendix \ref{ap:martingale}.
\end{proof}

\begin{lemma}\label{lm:dominate}
$\sum_{\ell\in\mathcal{L}_2}  \mathbf{v}_{\ell}(\mathbf{G}^*_{\ell}-\mathbf{G}_{\ell})\mathbf{h}_{\ell} =O\Big(V^{2+N\slash 2} \sqrt{T\log(T V)}\Big)$.
\end{lemma}
\begin{proof}
See Appendix \ref{ap:dominate}.
\end{proof}

\vspace{2mm}

As a result, combining Lemma \ref{lm:martingale} and Lemma \ref{lm:dominate}, we can conclude that with probability at least $1-\frac{1}{T}$
\[
\sum_{\ell\in\mathcal{L}_2}\sum_{(\mathbf{Q},\alpha)} v_{\ell}(\mathbf{Q},\alpha) \Big(\sum_{i,k} Q_{ik} - J_{\ell}\Big) = O\Big(V^{N+2}\log T + V^{2+N\slash 2}\sqrt{ T\log(TV)}\Big),
\]
which completes the proof.

\section{Proof to Lemma \ref{lm:martingale}}\label{ap:martingale}
We find an upper bound for $\sum_{\ell\in\mathcal{L}_2} \mathbf{v}_{\ell}(\mathbf{I}-\mathbf{G}^*_{\ell})\mathbf{h}_{\ell}$. We first define a suitable martingale and make use of the Azuma-Hoeffding inequality.

\begin{lemma}[Azuma-Hoeffding inequality]\label{lm:azuma}
Let $X_1,X_2,\cdots$ be a martingale difference sequence with $|X_i|\le c$ for all $i$. Then for any $\epsilon>0$ and $n\in\mathbb{N}$,
\[
\mathbb{P}\Big[\sum_{i=1}^{n} X_i \ge \epsilon \Big] \le \exp\Big(-\frac{\epsilon^2}{2nc^2}\Big).
\]
\end{lemma}
Define a $|\mathcal{Q}_V|$-dimensional vector $\mathbf{e}_{Q}$ where the coordinate corresponding to state $\mathbf{Q}$ being 1 and all other coordinates being 0. Let $\mathcal{F}_t=\big(\mathbf{Q}(0),\alpha_0,\mathbf{Q}(1),\cdots,\mathbf{Q}(t),\alpha_{t}\big)$ be the sequence of states and actions up to slot $t$, and let $\ell(t)$ be the episode containing slot $t$. 
Consider the sequence 
\[
X_t\triangleq \Big(\mathbf{e}_{Q(t+1)}-P^*(\cdot | \mathbf{Q}(t), \alpha_{t})\Big)\mathbf{h}_{\ell(t)},~t=0,1\cdots,T,
\]
where $\alpha_t=\pi_{\ell}(\mathbf{Q}(t))$ is the action taken by the TUCRL algorithm in slot $t$. Note that conditioned on the historical sequence $\mathcal{F}_t$, the randomness of $X_t$ is in the next state $\mathbf{Q}(t+1)$. Also, $\{X_t\}_{t\ge 0}$ is a sequence of martingale differences since $\mathbb{E}_{Q(t+1)}[X_t|\mathcal{F}_{t}]=0$ for all $t$. In addition, we have $|X_t|\le ||\mathbf{e}_{Q(t+1)}\mathbf{h}_{\ell}||_{1} \le ||\mathbf{h}_{\ell}||_{\infty} \le LN^2V^2$. 

Then we have
\[
\begin{split}
\mathbf{v}_{\ell}(\mathbf{I}-\mathbf{G}^*_{\ell})\mathbf{h}_{\ell} &= \sum_{t=t_{\ell}}^{t_{\ell+1}-1}\Big(\mathbf{e}_{Q(t)}-P^*(\cdot | \mathbf{Q}(t), \alpha_{t})\Big)\mathbf{h}_{\ell}\\
&=\Big(\sum_{t=t_{\ell}}^{t_{\ell+1}-1} \mathbf{e}_{Q(t+1)}-\sum_{t=t_{\ell}}^{t_{\ell+1}-1}P^*(\cdot | \mathbf{Q}(t), \alpha_{t}) -\mathbf{e}_{Q(t_{\ell+1})}+\mathbf{e}_{Q(t_{\ell})}\Big)\mathbf{h}_{\ell}\\
&= \sum_{t=t_{\ell}}^{t_{\ell+1}-1} X_t + h_{\ell}(\mathbf{Q}(t_{\ell}))-h_{\ell}(\mathbf{Q}(t_{\ell+1}))\\
&\le \sum_{t=t_{\ell}}^{t_{\ell+1}-1} X_t + LN^2V^2,
\end{split}
\]
where the last inequality is due to $||\mathbf{h}_{\ell}||_{\infty}\le LN^2V^2$. Let $\mathcal{T}_2$ be the set of slots contained in episodes $\ell\in\mathcal{L}_2$. By the Azuma-Hoeffding inequality, we have
\[
\mathbb{P}\Big[\sum_{t\in\mathcal{T}_2} X_t \ge LN^2V^2\sqrt{2|\mathcal{T}_2|\log T}\Big]\le \frac{1}{T}.
\]
Summing over $\ell\in\mathcal{L}_2$, we have  with probability at least $1-\frac{1}{T}$
\[
\begin{split}
\sum_{\ell\in\mathcal{L}_2} \mathbf{v}_{\ell}(\mathbf{I}-\mathbf{G}^*_{\ell})\mathbf{h}_{\ell} &\le \sum_{t\in\mathcal{T}_2} X_t + LN^2V^2|\mathcal{L}_2|\\
&\le LN^2V^2\sqrt{2|\mathcal{T}_2|\log T} + LN^2V^2\Big(1+V^N|\mathcal{A}|(\log T+1)\Big)\\
&\le  O(V^2\sqrt{T\log T}+V^{N+2}\log T),
\end{split}
\]
where we use the fact that $|\mathcal{T}_2|\le T$ and $|\mathcal{L}_2|\le 1+|\mathcal{Q}_V||\mathcal{A}|(\log_2 T+1)$ (see Lemma \ref{lm:episode}).

\section{Proof to Lemma \ref{lm:dominate}}\label{ap:dominate}
We find an upper bound for $\sum_{\ell\in\mathcal{L}_2} \mathbf{v}_{\ell}(\mathbf{G}^*_{\ell}-\mathbf{G}_{\ell})\mathbf{h}_{\ell}$. Since both the truncated true MDP $M_V^*$ and the selected optimistic MDP $M_{\ell}$ are in the confidence set $\mathcal{M}_{\ell}$, the term $\mathbf{G}^*_{\ell}-\mathbf{G}_{\ell}$ can be bounded using the confidence set construction \eqref{eq:confidence-set}. We have
\begin{equation}
\begin{split}
\mathbf{v}_{\ell}(\mathbf{G}^*_{\ell}-\mathbf{G}_{\ell})\mathbf{h}_{\ell} & = \sum_{\mathbf{Q} \in\mathcal{Q}_V} v_{\ell}\Big(\mathbf{Q},\pi_{\ell}(\mathbf{Q})\Big)\sum_{\mathbf{Q}'\in \mathcal{Q}_V}\Big(P_V^*\big(\mathbf{Q}'|\mathbf{Q}, \pi_{\ell}(\mathbf{Q})\big)-P_{\ell}\big(\mathbf{Q}'|\mathbf{Q},\pi_{\ell}(\mathbf{Q})\big)\Big)h_{\ell}(\mathbf{Q}')\\
&\le \sum_{(\mathbf{Q},\alpha)} v_{\ell}(\mathbf{Q},\alpha)\sqrt{\frac{C \log(2|\mathcal{A}|t_{\ell} V)}{\max\{1,n_{\ell}(\mathbf{Q},\alpha)\}}} ||\mathbf{h}_{\ell}||_{\infty}\\
&\le \sqrt{C}LN^2\cdot \sum_{(\mathbf{Q},\alpha)} v_{\ell}(\mathbf{Q},\alpha)\sqrt{\frac{\log(2|\mathcal{A}|T V)}{\max\{1,n_{\ell}(\mathbf{Q},\alpha)\}}} V^2,
\end{split}
\end{equation}
where $P_V^*$ is the transition matrix for the truncated true MDP $M_V^*$ and we use the fact that $||\mathbf{h}_{\ell}||_{\infty}\le LN^2V^2$. 

Let $n(\mathbf{Q},\alpha)=\sum_{\ell} v_{\ell}(\mathbf{Q},\alpha)$ such that $\sum_{(\mathbf{Q},\alpha)} n(\mathbf{Q},\alpha)=T$. Moreover, recall that $n_{\ell}(\mathbf{Q},\alpha)=\sum_{i<\ell} v_i(\mathbf{Q},\alpha)$, and that $v_{\ell}(\mathbf{Q},\alpha)\le n_{\ell}(\mathbf{Q},\alpha)$ by the stopping condition of episode $\ell$. As a result, we have
\[
\begin{split}
\sum_{\ell} \mathbf{v}_{\ell}(\mathbf{G}^*_{\ell}-\mathbf{G}_{\ell})\mathbf{h}_{\ell} & \le  \sqrt{C}LN^2\cdot \sum_\ell \sum_{(\mathbf{Q},\alpha)} v_{\ell}(\mathbf{Q},\alpha)\sqrt{\frac{\log(2|\mathcal{A}|T V)}{\max\{1,n_{\ell}(\mathbf{Q},\alpha)\}}} V^2\\
& =  \sqrt{C}LN^2 V^2 \sqrt{\log(2|\mathcal{A}|T V)}  \sum_{(\mathbf{Q},\alpha)} \sum_{\ell} \frac{v_{\ell}(\mathbf{Q}, \alpha)}{\sqrt{\max\{1,n_{\ell}(\mathbf{Q},\alpha)\}}}\\
&\le  \sqrt{C}LN^2V^2 (\sqrt{2}+1)  \sqrt{\log(2|\mathcal{A}|T V)}  \sum_{(\mathbf{Q},\alpha)} \sqrt{n(\mathbf{Q},\alpha)}\\
& \le  \sqrt{C}LN^2V^2 (\sqrt{2}+1)  \sqrt{\log(2|\mathcal{A}|T V)} \sqrt{V^N |\mathcal{A}| T}.
\end{split}
\]
Here we used for the second inequality that
\[
\sum_{\ell=1}^n \frac{x_\ell}{\sqrt{X_{\ell-1}}}\le (\sqrt{2}+1)\sqrt{X_n}.
\]
where $X_\ell=\max\{1,\sum_{i=1}^\ell x_i\}$ and $0\le  x_\ell\le X_{\ell-1}$ (see Lemma 19 in \cite{UCRL2}), and we used Jensen's inequality for the last inequality. Finally, by the definition of $\mathcal{L}_1$ and $\mathcal{L}_2$ we have
\[
\begin{split}
\sum_{\ell\in\mathcal{L}_2} \mathbf{v}_{\ell}(\mathbf{G}^*_{\ell}-\mathbf{G}_{\ell})\mathbf{h}_{\ell}&=\sum_{\ell}\mathbf{v}_{\ell}(\mathbf{G}^*_{\ell}-\mathbf{G}_{\ell})\mathbf{h}_{\ell} +  \sum_{\ell\in\mathcal{L}_1}\mathbf{v}_{\ell}(\mathbf{G}_{\ell}-\mathbf{G}^*_{\ell})\mathbf{h}_{\ell}\\
&\le \sum_{\ell} \mathbf{v}_{\ell}(\mathbf{G}^*_{\ell}-\mathbf{G}_{\ell})\mathbf{h}_{\ell} + \sqrt{C}LN^2V^2 \log T \sqrt{\log(2|\mathcal{A}|TV)}\\
& \le \sqrt{C}LN^2V^2 (\sqrt{2}+1)  \sqrt{\log(2|\mathcal{A}|T V)} \sqrt{V^N |\mathcal{A}| T}+ \sqrt{C}LN^2V^2 \log T \sqrt{\log(2|\mathcal{A}|TV)}\\
&=O\Big(V^{2+N\slash 2} \sqrt{T\log(T V)}\Big),
\end{split}
\]
which completes the proof.

\section{Proof to Theorem \ref{thm:drop}}\label{ap:drop}
Let $\mathcal{X}$ be the set of slots when $\sum_{i,k} Q_{ik}(t) \ge V-ND$. Note that the amount of exogenous arrivals to each queue is at most $D$ in each time slot, and thus packet dropping occurs in a  slot $t$ only if $t\in\mathcal{X}$. As a result, in order to bound the amount of dropped packets, we find an upper bound for $|\mathcal{X}|$. 

Note that
\[
\sum_{t=0}^{T-1}\sum_{i,k} Q_{ik}(t) \ge \sum_{t\in\mathcal{X}} \sum_{i,k} Q_{ik}(t) \ge (V-ND) |\mathcal{X}|.
\]
By Theorem \ref{thm:truncated-bound}, with probability at least $1-O\Big(\frac{1}{V}+\frac{1}{T}\Big)$, the TUCRL algorithm achieves
\[
\begin{split}
\sum_{t=0}^{T-1}\sum_{i,k} Q_{ik}(t) &\le TJ_V^* + \gamma(T,V),
\end{split}
\]
which implies that
\begin{equation}\label{eq:x}
|\mathcal{X}| \le  \frac{\gamma(T,V) + TJ^*_V}{V-ND}.
\end{equation}
Since at most $ND$ packets are dropped for each slot $t\in\mathcal{X}$, the amount of dropped packets within $T$ slots is at most 
\[
ND|\mathcal{X}|\le \frac{(\gamma(T,V) + TJ^*_V)ND}{V-ND}.
\]
Note that the amount of total exogenous arrivals $\sum_{i,k} a_{ik}(t)$ in each slot $t$ is i.i.d. with expectation $\mathbb{E}[\sum_{i,k} a_{ik}(t)]=\sum_{i,k} \lambda_{ik}$, and  $\sum_{i,k} a_{ik}(t)\in [0,ND]$. By Hoefdding's inequality, we have
\[
\mathbb{P}\Big[\sum_{t=0}^{T-1}\sum_{i,k} a_{ik}(t) \le \frac{T\sum_{i,k} \lambda_{ik}}{2}\Big]\le \exp\Big(-\frac{T(\sum_{i,k}\lambda_{ik})^2}{2N^2D^2}\Big).
\]
When $T> \frac{2N^2D^2\log T}{(\sum_{i,k} \lambda_{ik})^2}$, the above probability is less than $\frac{1}{T}$.
As a result, with probability at least $1-O\Big(\frac{1}{T}+\frac{1}{V}\Big)$, the fraction of dropped packets is at most
\[
\frac{ND|\mathcal{X}|}{\sum_{t=0}^{T-1}\sum_{i,k} a_{ik}(t)}\le  \frac{2(\gamma(T,V)+TJ_V^*)ND}{T(V-ND)\sum_i \lambda_i}=O\Big(\frac{\gamma(T,V)}{TV}+\frac{J_V^*}{V}\Big).
\]

\section{Proof to Lemma \ref{lm:cc}}\label{ap:cc}
By Assumption \ref{as:mdp-stability}, we have that for any $t\ge 0$
\[
\mathbb{P}\Big(\sum_{i,k} Q_{ik}^*(t)\le C\Big)=1,
\]
for some constant $C>0$. Then it follows that for any $T\ge 0$
\[
\mathbb{P}\Big(\max_{t\le T}\sum_{i,k} Q_{ik}^*(t)\le C\Big) = 1.
\]
By Bounded Convergence Theorem, we have
\[
\mathbb{P}\Big(\lim_{T\rightarrow\infty}\max_{t\le T}\sum_{i,k} Q_{ik}^*(t)\le C\Big) =\lim_{T\rightarrow\infty}\mathbb{P}\Big(\max_{t\le T}\sum_{i,k} Q_{ik}^*(t)\le C\Big)= 1
\]
As a result, we have
\[
\mathbb{E}\Big[\max_{t\ge 0} \sum_{i,k} Q^*_{ik}(t)\Big]\le C.
\]
By Markov inequality, we have
\[
\mathbb{P}[\max_{t\ge 0} \sum_{i,k} Q_{ik}^*(t) \ge V] \le \frac{\mathbb{E}\Big[\max_{t\ge 0} \sum_{i,k} Q^*_{ik}(t)\Big]}{V}\le \frac{C}{V}=O\Big(\frac{1}{V}\Big).
\]
 As a result, with probability at least $1-O\Big(\frac{1}{V}\Big)$, the optimal policy never visits any queue length vector outside $\mathcal{Q}_V$ and the optimal average queue length is not affected in the truncated MDP, i.e., $J^*_V=J^*<\infty$, which implies that $\mathbb{P}[J_V^*=J^*]\ge 1-O\Big(\frac{1}{V}\Big)$.

\end{document}